\documentclass[conference]{IEEEtran}
\IEEEoverridecommandlockouts

\usepackage{amsmath,amsfonts,amssymb}
\usepackage{amsthm}
\usepackage{array}
\usepackage[caption=false,font=normalsize,labelfont=sf,textfont=sf]{subfig}
\usepackage[bottom=0.85in,top=0.71in,left=0.65in,right=0.65in]{geometry}
\usepackage{textcomp}
\usepackage{stfloats}
\usepackage{url}
\usepackage{xcolor}
\usepackage{verbatim}
\usepackage{graphicx}
\usepackage{bm}
\usepackage{cite}
\usepackage{footnote}
\usepackage{caption}
\usepackage{dsfont}
\usepackage{booktabs}
\usepackage{multirow}

\usepackage[ruled,vlined]{algorithm2e}
\SetArgSty{textnormal}

\newtheorem{theorem}{Theorem}
\newtheorem{lemma}{Lemma}
\newtheorem{remark}{Remark}
\newtheorem{Prob}{\textbf{Problem}}

\title{LLM-Guided Dual-Stage Latency Optimization for Closed-Loop 6G V2I Digital Twins}

 \author{ \text{Liu Cao},~\IEEEmembership{Member,~IEEE}, \text{Peide Li}, \text{Yuan Yang},  \text{Zhaoyu Liu},  \text{Lyutianyang Zhang},~\IEEEmembership{Member,~IEEE},\\\text{Jinke Ren},~\IEEEmembership{Member,~IEEE}, and \text{Shuguang Cui},~\IEEEmembership{Fellow,~IEEE}

 \thanks{Liu Cao, Peide Li, Yuan Yang, and Zhaoyu Liu are with Electrical and Electronic Engineering Program of both City University of Hong Kong (Dongguan), Dongguan, Guangdong 523808, China and City University of Hong Kong, Hong Kong (e-mail:liu.cao@cityu-dg.edu.cn).  
 
 Lyutianyang Zhang is  with the School of Microelectronics and Communication Engineering, Chongqing
University, Chongqing 400044, China (email: zhanglyutianyang@cqu.edu.cn).

Jinke Ren and Shuguang Cui are with the School of Science and Engineering, the Shenzhen Future Network of Intelligence Institute, and the Guangdong Provincial Key Laboratory of Future Networks of Intelligence, The Chinese University of Hong Kong (Shenzhen), Shenzhen, Guangdong
518172, China (email: \{jinkeren, shuguangcui\}@cuhk.edu.cn).  (Corresponding author: Lyutianyang Zhang.)

}

}

\begin{document}

\maketitle

\begin{abstract}
Wireless Digital Twins (DTs) are emerging as a key enabler for predictive network control in 6G Vehicle-to-Infrastructure (V2I) systems. However,  closed-loop V2I DTs face two dual-stage communication bottlenecks arising from model parameter exchange for V2I DT reconstruction and multi-service V2I traffic delivery after each reconstruction. This paper proposes a closed-loop dual-stage End-to-End (E2E)  latency optimization framework that integrates Compressed Personalized Federated Learning (CPFL) with a unified Large-Language-Model-guided Multi-Agent Reinforcement Learning (LLM-MARL) approach. A CPFL-enabled decision-oriented V2I DT is first constructed by jointly predicting Path Loss (PL) and Quality-of-Experience (QoE)-oriented Key Performance Indicators (KPIs). Then the framework comprises two temporally coupled but latency-wise distinct optimization stages: Stage I minimizes the joint Uplink (UL) and Downlink (DL) E2E transmission latency of compressed shared backbone's parameters, and Stage II minimizes the service-aware E2E latency of multi-service V2I traffic. As the two stages share aligned agent roles and action space, we develop a unified Semantic-Aware Guidance-Enhanced Multi-Agent Proximal Policy Optimization (SAGE-MAPPO) framework, in which the LLM acts as a training-time semantic planner that provides high-level traffic splitting and power allocation guidance to reduce inefficient exploration, while the MARL agents execute feasible low-level actions. Simulation results demonstrate that the proposed framework enables efficient V2I DT reconstruction and significantly reduces both E2E latency for DT reconstruction and multi-service V2I traffic. The proposed SAGE-MAPPO supports cross-stage policy reuse by sequentially solving the two stages within each DT update cycle, which accelerates convergence relative to the existing baselines and facilitates timely DT updates in dynamic V2I scenarios.

\end{abstract}



\section{Introduction}

\IEEEPARstart{W}{ith} the rapid evolution toward 6G intelligent transportation
systems, Vehicle-to-Infrastructure (V2I) networks are expected to support highly dynamic, latency-sensitive, and service-aware communications\cite{10438479,10945654}. In dense urban environments, vehicles experience fast channel fluctuation, frequent Line-of-Sight (LoS)/Non-LoS (NLoS) transitions, dynamic blockage, and time-varying traffic loads. These characteristics make it difficult for conventional model-based network planning tools to provide accurate and timely guidance for V2I resource management\cite{Iyer_2022,11063460}. As a result, wireless Digital Twins (DTs) have emerged as a promising solution for building a synchronized virtual representation of the physical wireless environment, enabling predictive network control and closed-loop optimization \cite{li2026digital}.

However, constructing a practical V2I DT is generally challenging. First, a DT should provide decision-oriented network states, such as Signal-to-Interference-plus-Noise Ratio (SINR), achievable data rate, Base Station (BS) load, queueing delay, and Quality-of-Experience (QoE)-related Key Performance Indicators (KPIs). These quantities are directly coupled with downstream V2I association and latency-related optimization. Second, the local measurements collected by different vehicles are naturally Non-Independent and Identically Distributed (Non-IID) due to heterogeneous locations, mobility patterns, blockage conditions, and traffic demands \cite{kairouz2021advancesopenproblemsfederated}. Third, the transmission latency for Federated Learning (FL) is however highly dependent on the link quality of V2I association. Poor V2I association may increase both Uplink (UL) and Downlink (DL) latency, making FL training itself a significant bottleneck for efficient DT reconstruction \cite{arivazhagan2019federated}. 

In addition to the demand on efficient DT reconstruction, the DT should be able to serve End-to-End (E2E) V2I service. In practical 6G V2I systems, vehicles are able to transmit multiple Vehicle-to-Everything (V2X) types of traffic simultaneously through associated BSs \cite{3gpp.38.885,Garcia2021NRV2X,cao2022resource}, and the resulting E2E latency consists of both UL and DL components. Optimizing this latency requires the joint control of UL traffic splitting, transmit power, and DL traffic splitting. This problem is difficult because the action space is continuous, the interference and queueing dynamics are coupled, and the latency objective is constrained by the slowest active link path. While Multi-Agent Reinforcement Learning (MARL) is a promising tool for solving such dynamic resource allocation problems, however, conventional MARL algorithms often suffer from slow convergence and inefficient exploration in large-scale 6G V2X networks.

\subsection{Related Work}
Recently, wireless DTs have been regarded as a key enabler for 6G network modeling, predictive control, and closed-loop optimization~\cite{lin2023digitaltwin,tao2024wirelessdt}. Existing wireless DT studies emphasize on synchronizing physical wireless environments with digital replicas for online analysis and network control. Different from these studies, our work develops a compressed FL framework to construct a decision-oriented V2I DT. Meanwhile, different from centralized training~\cite{10474098}, FL addresses privacy concerns by keeping raw data local~\cite{dai2020deeplearningwirelesscommunications,niknam2020federatedlearningwirelesscommunications,11096953} but struggles with data heterogeneity~\cite{Lan2023}, communication bottlenecks~\cite{WU2025,zhang2025cross,ren2020scheduling}, and multi-task complexity. Personalized FL (PFL) mitigates Non-IID effects by retaining client-specific parameters alongside a shared global representation~\cite{fallah2020personalized,shome2022federatedlearninggenerationwireless,liu2024surveyfederatedunlearningchallenges}. Furthermore, gradient compression techniques are widely used to reduce communication overhead~\cite{lin2018deepgradient}. However, most existing FL methods only reduce UL communication overhead and do not explicitly optimize the joint UL and DL parameter-transmission latency in each communication round. 

Traditional Path Loss (PL) prediction methods mainly rely on parameterized empirical models or ray-tracing simulation, which face computational bottlenecks and generalization limits in complex 3D urban environments~\cite{10347228,albagami2024universaldeepneuralnetwork}. Consequently, data-driven approaches using deep neural networks, such as convolutional neural networks (CNNs) and graph neural networks (GNNs), have become the mainstream solutions for radio map construction by learning non-linear mappings between environmental features and signal fading~\cite{levie2021radiounet,11048575}. However, most existing models treat environmental inputs as static binary masks and lack explicit modeling of geometric  relationships, such as LoS and NLoS propagation paths. While accurate PL mapping is fundamental, network operations depend directly on KPIs, such as throughput, queueing delay, load, and fairness. Prior data-driven KPI prediction methods \cite{vishnu2023data} often separate radio map construction from multi-BS association decisions, neglecting the non-linear queueing effects and load-coupling dynamics induced by traffic distribution. Multi-Task Learning (MTL) provides a natural framework to jointly optimize link-level PL and user-level KPIs by sharing feature representations~\cite{qin2020federatedlearningwirelesscommunications}. In contrast, our work jointly learns PL and KPI-related states, including SINR, data rate, Micro-BS (MBS) load, queueing delay, and QoE-related metrics. Therefore, the reconstructed DT can directly support V2I association and latency optimization. Meanwhile, traditional hard association strategies, such as Max-SINR, frequently cause severe load imbalance in multi-BS networks~\cite{10662910}. To alleviate this issue, soft association instead distributes a vehicle's traffic over multiple MBSs using normalized rate- or utility-proportional weights whose sum is one~\cite{11304140,TRINDADE2024765}. However, in dense urban V2I environments, instantaneous channel quality alone is insufficient because high-rate links may still suffer from unstable NLoS blockage or severe queueing congestion. 

MARL has been widely adopted for wireless resource allocation because it can learn distributed policies under coupled interference, traffic, and queueing dynamics~\cite{du2023marl6g,yu2022mappo}. Nevertheless, standard MARL often requires extensive exploration when the joint action space includes traffic splitting, power control, and multi-BS scheduling. Recently, Large Language Models (LLMs) have been used as high-level reasoning modules for wireless resource allocation~\cite{lee2024llmresource}. However, directly using an LLM as a low-level controller is computationally expensive and may violate strict feasibility constraints. In this work, the LLM is not used to directly output final control decisions. Instead, it provides semantic guidance and soft policy priors for a unified MARL framework, helping the MARL agents identify overloaded MBSs, unreliable links, and preferred traffic splitting directions.

\subsection{Motivation and Contributions}
To address these challenges, we propose a closed-loop framework that integrates periodic V2I DT reconstruction with a unified LLM-guided MARL framework to solve the dual-stage E2E latency optimization problems, where the joint E2E UL and DL transmission latency and V2X service-aware traffic are iteratively optimized respectively. A key feature of our design is that  {\em the LLM acts as a semantic planner that summarizes global network states and provides high-level guidance, such as identifying overloaded BSs, unreliable links, and preferred traffic splitting tendencies. The final feasible actions are still executed by the proposed MARL agents and optimized through reward feedback. This design preserves the adaptability and constraint awareness of MARL while using LLM-based semantic reasoning to reduce redundant exploration and accelerate convergence.}

The main contributions are summarized as follows:

\begin{itemize}
\item We develop a Compressed Personalized FL (CPFL)-enabled DL V2I DT, augmented with online UL channel measurements. The resulting hybrid state provides a consistent decision basis for bidirectional V2I transmission control. Meanwhile, we design a joint geometry and KPI-aware V2I association strategy assisted by the reconstructed DT for downstream E2E multi-path transmission architecture.
\item We formulate dual E2E latency optimization problems in which transmission structures occur in a unified manner in each closed-loop DT-update cycle. Both structures share the same agent roles and action space where each vehicle optimizes the UL traffic splitting and transmit power, while the edge server controller optimizes the DL traffic splitting.
\item We develop a unified Semantic-Aware Guidance-Enhanced Multi-Agent Proximal Policy Optimization (SAGE-MAPPO) framework to solve both optimization problems. The LLM provides training-time semantic guidance, helping MARL agents reduce inefficient exploration. By sharing the same architecture across both E2E latency optimization problems, the proposed method enables cross-problem policy reuse and achieves faster convergence than training-independent MARL controllers.
\end{itemize}

The rest of this paper is organized as follows. 
Sec. II introduces the system model and the closed-loop CPFL-enabled V2I DT reconstruction framework. 
Sec. III presents the unified LLM-MARL framework. Sec. IV provides the performance evaluation of DT reconstruction, convergence behavior, and E2E latency optimization. 
Finally, Sec. V concludes this paper.

\section{Dual-Stage Optimization Problem Formulation}

\begin{figure}[!t]
\centering
\includegraphics[width=0.48\textwidth]{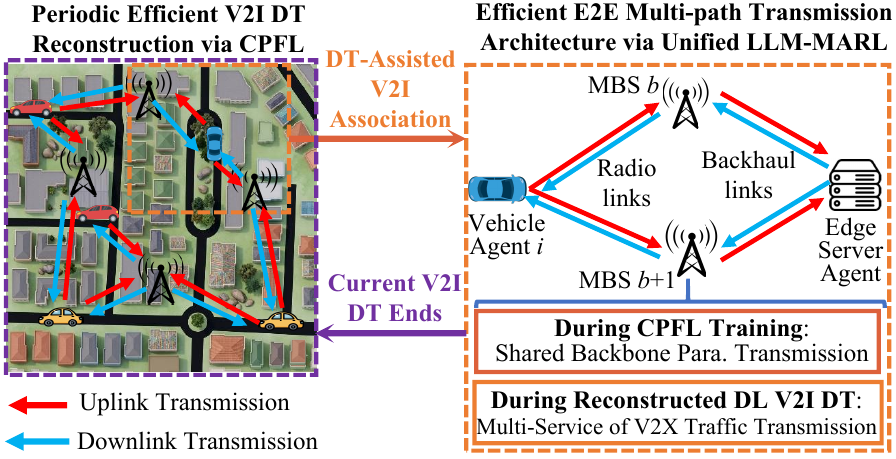} 
\caption{Closed-Loop 6G V2I DT-Assisted Dual-Stage E2E Latency
Optimization Framework.}
\label{fig_fedcompre}
\end{figure}


Let us consider a 6G V2X network consisting of a set of vehicles 
$\mathcal I=\{1,\ldots,I\}$, a set of MBSs $\mathcal B=\{1,\ldots,B\}$, and an edge server. Each vehicle collects local wireless measurements and environmental observations to collaboratively reconstruct a DL V2I DT. Since the wireless environment is time-varying due to vehicle mobility, dynamic blockage, channel fluctuation, and traffic-load variation, the V2I DT is periodically updated. Let $\mathcal K=\{1,\ldots,K\}$ denote the set of DT update indices, where each $k\in\mathcal K$ corresponds to one DT reconstruction cycle. Let $\mathcal D_i^{(k)}$ denote the local dataset measured at vehicle $i$ during the $k$-th DT update cycle. To reduce communication overhead while preserving personalization, each vehicle employs a CPFL model consisting of a global shared backbone $\mathbf{w}^{(k)}$ and a local head $\mathbf{v}_i^{(k)}$. The shared backbone captures common wireless propagation and traffic patterns, while the personalized head captures the DL channel, blockage, and load characteristics. During the CPFL training, only the compressed shared backbone parameters are transmitted between the server and vehicles, while the local head parameters are not transmitted. We assume that each vehicle $i$ can obtain its DL local dataset $\mathcal D_i^{(k)}$ through practical V2I measurements from the MBS set $\mathcal B$, including link-quality, traffic-load, and environmental observations. The raw measured data remain local and are not uploaded during CPFL training.

\subsection{Multi-task DT Prediction}
The proposed CPFL model performs two coupled DT tasks:
\subsubsection{DT PL Prediction}


The first task of the CPFL-enabled DT is PL prediction. The PL from MBS $b$ to vehicle $i$ is given by $PL_{b,i}$, which can be obtained based on the measured received power at vehicle $i$ and the transmit power obtained from the control information in Physical Downlink Control Channel (PDCCH). The PL prediction head learns the nonlinear mapping $f_i^{\mathrm{PL}}$ from local DT features to the measured PL vector:
\begin{equation}
\widehat{\mathbf y}_{i}^{\mathrm{PL}}
=
f_i^{\mathrm{PL}}
\left(
\mathbf z_i;
\mathbf w,
\mathbf v_i
\right)
\approx
\left\{
PL_{b,i}(x_{\text{loc}},y_{\text{loc}})
\right\}_{b\in\mathcal B},
\label{eq:stage1_pl_prediction}
\end{equation}
where $\widehat{\mathbf y}_{i}^{\mathrm{PL}}$ is the PL vector predicted by vehicle $i$. $\mathbf z_i$ denotes the geometry-enhanced local DT feature vector observed by vehicle $i$ at its location ($x_{\text{loc}}$, $y_{\text{loc}}$), which is constructed as $\mathbf z_i
=$ $[
x_{\text{loc}}, 
y_{\text{loc}}, 
\text{building\ mask}, 
\text{road\ mask}, \text{vehicle\ mask}, $ $ \text{relative \ MBS\ distance\ features}, \text{LoS\ confidence\ features}]$.

\subsubsection{DT KPI Prediction}
The second task is DT KPI prediction. The DT KPI vector of vehicle $i$ is constructed as
\begin{equation}
\mathbf y_i^{\mathrm{KPI}}
=
\left[
\{\gamma_{b,i}^{\mathrm{DL}}\}_{b\in\mathcal B},
\{R_{b,i}^{\mathrm{DL}}\}_{b\in\mathcal B},
\left\{L_b\right\}_{b\in\mathcal B},\{D_b\}_{b\in\mathcal B},
J\right],
\label{eq:stage1_kpi_label}
\end{equation}
where $\gamma_{b,i}^{\mathrm{DL}}$ and $R_{b,i}^{\mathrm{DL}}$ are the measured SINR and DL data rate at vehicle $i$ from MBS $b$, respectively. $L_b$ is the data load of MBS $b$ given the association weight $w_{i,b}$ between vehicle $i$ and MBS $b$, which is given by
\begin{equation}
L_b
=
\sum_{i\in\mathcal I}
w_{i,b}.
\end{equation}
Note that $w_{i,b}$ depends on the V2I association strategy to be determined. $D_b$ is the expected queueing delay of MBS $b$ that follows the M/M/1 queueing model, which is expressed as
\begin{equation}
D_b=\frac{1}{\mu_b-\lambda_b},
\label{eq:eval_queue_delay}
\end{equation}
where $\mu_b$ is the service rate and $\lambda_b$ is the traffic arrival rate at MBS $b$.
$J$ is the Jain's fairness index among all vehicles, which is expressed as
\begin{equation}
J
=
\frac{
\left(
\sum_{i\in\mathcal I}R_i^{DL}
\right)^2
}{
|\mathcal I|
\sum_{i\in\mathcal I}(R_i^{DL})^2
},
\end{equation} 
where the DL data rate at vehicle $i$ is given by
\begin{equation}
R_i^{\mathrm{DL}}
=
\sum_{b\in\mathcal B_c}
w_{b,i}
R_{b,i}^{\mathrm{DL}},
\end{equation}
where $R_{b,i}^{\mathrm{DL}}$ is the DL data rate at vehicle $i$ from MBS $b$. The KPI prediction head learns the mapping $\widehat{\mathbf y}_{i}$ from local DT features to the measured DT KPI vector:
\begin{equation}
\widehat{\mathbf y}_{i}^{\mathrm{KPI}}
=
f_i^{\mathrm{KPI}}
\left(
\mathbf z_i;
\mathbf w,
\mathbf v_i^{\mathrm{KPI}}
\right)
\approx
\mathbf y_i^{\mathrm{KPI}},
\label{eq:stage1_kpi_prediction}
\end{equation}
where $\widehat{\mathbf y}_{i}^{\mathrm{KPI}}$ is the KPI vector predicted by vehicle $i$. 

We evaluate the DL performance between vehicle $i$ and each MBS $b$ in terms
of the predicted KPI vector. Let $V_{i,b}$ denote a KPI-aware effective V2I
metric between vehicle $i$ and MBS $b$ that jointly reflects the link-level and system-level quality, which is given by
\begin{equation}
V_{i,b}
=
\boldsymbol{\omega}^{\mathrm T}
g_{i,b}(\widehat{\mathbf y}_{i}^{\mathrm{KPI}}),
\label{eq:kpi_aware_vehicle_score}
\end{equation}
where $g_{i,b}(\widehat{\mathbf y}_{i}^{\mathrm{KPI}})$ is the normalized KPI  vector as a function of the predicted KPI vector $\widehat{\mathbf y}_{i}^{\mathrm{KPI}}$, which is expressed as
\begin{equation}
g_{i,b}(\widehat{\mathbf y}_{i}^{\mathrm{KPI}})
=
\left[
\widetilde{\gamma}_{b,i}^{\mathrm{DL}},
\widetilde{R}_{b,i}^{\mathrm{DL}},
1-\widetilde{L}_b,
1-\widetilde{D}_b,
\widetilde{J}
\right]^{\mathrm T},
\label{eq:normalized_kpi_vector}
\end{equation}
and $\boldsymbol{\omega}$ is the corresponding KPI-weight vector satisfying
\begin{equation}
\boldsymbol{\omega}
=
\left[
\omega_1,
\omega_2,
\omega_3,
\omega_4,
\omega_5
\right]^{\mathrm T},
\quad
\omega_k\ge 0,\ \sum_{k=1}^{5}\omega_k=1.
\label{eq:kpi_weight_vector}
\end{equation}


It should be noted that the proposed CPFL model reconstructs a downlink-oriented V2I DT that determines the association and supplies the DL-side decision states. It does not predict the action-dependent instantaneous UL metrics which are updated online.

\subsection{DL V2I DT Reconstruction}

The general training procedures for V2I DT reconstruction are summarized as follows:

\subsubsection{DT-assisted V2I Association}
A typical V2I association
strategy is based on the highest
predicted SINR or the highest achievable data rate, which is however insufficient in dense urban scenarios: a radio link with high SINR and data rate may still suffer from unstable NLoS impacts due to buildings, vehicles, and street-canyon blockage, leading to unreliable service, unbalanced BS load, and increased queueing
delay. To address these issues, we propose a DT-assisted NLoS-aware V2I association strategy that combines the predicted data rate with explicit geometric visibility information $\mathbf z_i$. 
We define a continuous LoS confidence $\psi_{i,b}^{(k)}$ using the Bresenham's line
algorithm~\cite{bresenham1965algorithm}. Specifically, the line segment between MBS $b$ and vehicle $i$ during $k$th DT is discretized into grid cells. Let
$N_{\mathrm{tot}}^{(k)}(i,b)$ be the number of grid cells on this line segment and
$N_{\mathrm{block}}^{(k)}(i,b)$ be the number of grid cells occupied by blockage regions. Then, the LoS confidence is defined as
\begin{equation}
\psi_{i,b}^{(k)} = \max \left( 0, 1 - \eta \cdot \frac{N_{block}^{(k)}(i,b)}{N_{tot}^{(k)}(i,b)} \right),
\label{eq:los_confidence}
\end{equation}
where $\eta$ is an adjustment coefficient. $\psi_i \to 1$ for pure LoS, while $\psi_i \to 0$ for deep occlusion, thereby serving as a physically meaningful geometric prior. Then the proposed V2I association weight over all candidate MBSs is
\begin{equation}
w_{i,b}^{(k)}
=
\frac{(V_{i,b}^{(k)})^{\beta} (\psi_{i,b}^{(k)})^{1-\beta}}
{\sum\limits_{b \in \mathcal{B}} (V_{i,b}^{(k)})^{\beta} (\psi_{i,b}^{(k)})^{1-\beta}},
\label{eq:dt_assoc_weight}
\end{equation}
where $V_{i,b}$ is calculated in Eq. (\ref{eq:kpi_aware_vehicle_score}), $\beta\in[0,1]$ is the sensitivity weight to balance the KPI-aware performance and the link-level reliability.  
After obtaining the soft association weights $w_{i,b}$ over all candidate MBSs, the vehicle $i$ associate with the $K_s$ MBSs with the top-$K_s$ largest weights $w_{i,b}^{(k)}$. Let $\mathcal B_i^{(k)}$ denote the associated top $K_s$ MBSs for vehicle $i$ in the $k \geq1$ DT update. The final association weight is then re-normalized within the selected top-$K_s$ MBS set:
\begin{equation}
w_{i,b}^{(k)}
=
\begin{cases}
\dfrac{w_{i,b}^{(k)}}
{\sum\limits_{b\in\mathcal B_i^{(k)}} w_{i,b}^{(k)}},
& b\in\mathcal B_i^{(k)},\\[6pt]
0, & b\notin\mathcal B_i^{(k)},
\end{cases}
\label{eq:dt_assoc_weight1}
\end{equation}
where the V2I association weights satisfy $\sum_{b\in\mathcal B_i}w_{i,b}=1, \forall i$. Note that for the initial V2I association $k = 0$, each vehicle selects the highest-SINR MBS used for the UL and DL transmission during CPFL training:
\begin{equation}
\mathcal B_i^{(0)}
=
\underset{b}{\text{argmax}} \ \{\gamma_{b,i}^{\mathrm{DL}}\}_{b\in\mathcal B}.
\label{eq:initial_sinr_assoc}
\end{equation}
That is, the corresponding association weight $w_{i,b}^{(k)}=1$. The entire DT-assisted NLoS-Aware V2I association procedure is described by Algorithm~\ref{alg:decision}.

\begin{algorithm}[!t]
\caption{{DT-Assisted NLoS-Aware V2I Association Strategy}}
\label{alg:decision}
\KwIn{
Vehicle set $\mathcal I$, MBS set $\mathcal B$, DT update index $k$, 
predicted KPI-aware V2I metric $\{V_{i,b}^{(k)}\}$, 
DL SINR $\{\gamma_{b,i}^{\mathrm{DL}}\}$, 
geometric map $Z_{\mathrm{geom}}$, 
LoS adjustment coefficient $\eta$, 
sensitivity weight $\beta$, 
number of selected MBSs $K_s$
}

\KwOut{
Association sets $\{\mathcal B_i^{(k)}\}_{i\in\mathcal I}$ and final association matrix 
$W^{(k)}=\{w_{i,b}^{(k)}\}$
}

Initialize $W^{(k)}=\mathbf 0_{|\mathcal I|\times|\mathcal B|}$\;

\For{each vehicle $i\in\mathcal I$}{

    \If{$k=0$}{
        Select the initial MBS set $\mathcal B_i^{(0)}$ according to Eq.~\eqref{eq:initial_sinr_assoc}\; 
    }
    
    \Else{
        Extract vehicle location $x_i$\;
        
        \For{each MBS $b\in\mathcal B$}{
            
            Compute the LoS confidence $\psi_{i,b}^{(k)}$ according to Eq.~\eqref{eq:los_confidence}\;
        }
        
        Compute the DT-assisted V2I association weights according to Eq.~\eqref{eq:dt_assoc_weight} and Eq.~\eqref{eq:dt_assoc_weight1};
        
        Select the associated top-$K_s$ MBS set $\mathcal B_i^{(k)}$;
    }
}

\Return{$\{\mathcal B_i^{(k)}\}_{i\in\mathcal I}$, $W^{(k)}=\{w_{i,b}^{(k)}\}$}

\end{algorithm}

\subsubsection{CPFL Model Update}
In the proposed CPFL-enabled V2I DT framework, two prediction tasks of PL and KPI are designed to share the same global backbone while being predicted by separate task heads. This design is motivated by the fundamental differences between PL and KPI in terms of physical mechanisms, statistical characteristics, and decision-making roles: PL primarily reflects the link-layer channel attenuation, which is influenced by geometric blockages, LoS/NLoS conditions, and relative positions of MBSs, exhibiting local variability and strong interpretability. In contrast, KPI represents the system-level performance, which depends on the scheduled V2I association, interference coupling, traffic dynamics, and queueing, exhibiting highly nonlinear and globally coupled behavior. Modeling them separately allows the PL head to focus on capturing the link-level characteristics, while the KPI head leverages the shared backbone features and PL outputs to effectively predict system-level metrics. Based on the above two DT reconstruction tasks, the CPFL loss objective is formulated as

\begin{equation}
\min_{\mathbf{w}^{(k)},\{\mathbf{v}_i^{(k)}\}_{i\in\mathcal I}}
\sum_{i\in\mathcal I}
\frac{1}{|\mathcal{I}|}
\mathcal L_i(\mathbf{w}^{(k)},\mathbf{v}_i^{(k)}),
\label{eq:pfl_dt_objective}
\end{equation}
where the local multi-task loss of vehicle $i$ is
\begin{equation}
\mathcal L_i
=
\alpha
\mathcal L_i^{\mathrm{PL}}
+
(1-\alpha)
\mathcal L_i^{\mathrm{KPI}},
\label{eq:pfl_multitask_loss}
\end{equation} 
where $\mathcal L_i^{\mathrm{PL}}
=
\frac{1}{|\mathcal D_i^{(k)}|}
\sum_{d\in\mathcal D_i^{(k)}}
\left\|
\widehat{\mathbf y}_{i}^{\mathrm{PL}}(d)
-
\mathbf y_i^{\mathrm{PL}}(d)
\right\|_2^2$
and
$\mathcal L_i^{\mathrm{KPI}}
=
\frac{1}{|\mathcal D_i^{(k)}|}
\sum_{d\in\mathcal D_i^{(k)}}
\left\|
\widehat{\mathbf y}_{i}^{\mathrm{KPI}}(d)
-
\mathbf y_i^{\mathrm{KPI}}(d)
\right\|_2^2$
are the PL and KPI prediction loss of vehicle $i$, respectively. $\alpha \in [0,1]$ is the multi-task balancing weight between the two prediction loss. Then each vehicle computes its local update
\begin{equation}
\Delta_i^{(k)}(t)=\boldsymbol{\mathbf{w}}_i^{(k)}(t)-\boldsymbol{\mathbf{w}}^{(k)}(t),
\label{eq:update}
\end{equation}
where $\boldsymbol{\mathbf{w}}_i^{(k)}(t)$ is the locally updated backbone during the $k$th DT update. 
To improve the communication efficiency, we propose to compress such update before transmission via the joint sparsification and quantization function:
\begin{equation}
\widetilde{\Delta}_i^{(k)}(t) = \mathcal Q\Big( \Delta_i^{(k)}(t) \Big),
\label{eq:delta}
\end{equation}
where the $\mathcal Q$ function first selects the top $p_s$\% ratio of entries in $\Delta_i^{(k)}(t)$ with the largest magnitudes, while setting all remaining entries to zero. Then the $\mathcal Q$ function quantizes the retained entries into a $n_q$-bit representation. This function ensures that only the changes are transmitted, effectively reducing the communication cost. After receiving all compressed shared backbone parameters from participating vehicles, the edge server dequantizes the values to approximate the original magnitudes, and aggregates them by weight-averaging:
\begin{equation}
\mathbf w^{(k)}(t+1) = \mathbf w^{(k)}(t) + \frac{1}{|\mathcal I|} \sum_{i \in \mathcal I} \widetilde{\Delta}_i^{(k)}(t).
\label{eq:compressw}
\end{equation}

We then prove that our proposed CPFL converges to a stationary global model with the non-convex loss. For notational convenience, we omit the index of DT index $k$ and unless otherwise stated. To derive the expected convergence rate, we make following assumptions:
\begin{itemize}
    \item {\em Assumption 1:} We assume that $f_i(\mathbf{w}, \mathbf{v})$ is differentiable and $L$-smooth for any vehicle $i$, i.e.,
$\| \nabla f_i(x)-\nabla f_i(y)\|\le L\|x-y\|$, where $x=(\mathbf{w},\mathbf{v}), y=(\mathbf{w}',\mathbf{v}')$. Equivalently, $f_i(y) \leq f_i(x) + \langle \nabla f_i(x), y - x \rangle + \frac{L}{2} \|y - x\|^2$.
\item {\em Assumption 2:} We assume that the stochastic gradient $g_i(x;\mathcal D_i)$ is unbiased  with bounded variance given the local dataset $\mathcal D_i$ at vehicle $i$, i.e., $\mathbb{E}\!\left[g_i(x;\mathcal D_i)\right]=\nabla f_i(x)$, $\mathbb{E}\!\left[\|g_i(x;\mathcal D_i)-\nabla f_i(x)\|^2\right]\le \sigma^2$.
\item {\em Assumption 3:} We assume that there exists $F^*$ such that $F(x) \geq F^*$ for all $x$.
\item {\em Assumption 4:} We assume bounded local data heterogeneity where the gradient of shared backbone satisfies
$\frac{1}{|\mathcal I|}\sum_{i}^{|\mathcal I|}
\left\|
\nabla_\mathbf{w} f_i(x)-\nabla_\mathbf{w} F(x)
\right\|^2
\le \gamma^2$,
where $\nabla_\mathbf{w} F(x) = \frac{1}{|\mathcal I|}\sum_{i=1}^{|\mathcal I|} \nabla_\mathbf{w} f_i(x)$.
\item {\em Assumption 5:}
For the compression operator $\mathcal Q(\cdot)$, we have $\mathbb E\|\mathcal Q(\Delta_i(t))-\Delta_i(t)\|^2
\le
(1-\omega)\|\Delta_i(t)\|^2,\quad 0<\omega\le1$,
where $\omega$ is the compression ratio.

\end{itemize}

\begin{lemma}
Under Assumptions 1 and 2, if the step sizes satisfy $\eta_\mathbf{w}, \eta_\mathbf{v} \leq \frac{1}{L}$, then we have
\begin{equation}
\begin{aligned}
f_i\!\left(x_i^{t,e+1}\right)
\le &
f_i\!\left(x_i^{t,e}\right)
-
\frac{1}{2}
\Big(
\eta_\mathbf{w} \bigl\|\nabla_\mathbf{w} f_i(x_i^{t,e})\bigr\|^2
\\&+
\eta_\mathbf{v} \bigl\|\nabla_\mathbf{v} f_i(x_i^{t,e})\bigr\|^2
\Big) 
+
\frac{L}{2}
\left(\eta_\mathbf{w}^2 + \eta_\mathbf{v}^2\right)\sigma^2 ,
\end{aligned}
\end{equation}
where $x_i^{t,e} \triangleq (\mathbf{w}_i^{t,e},\mathbf{v}_i^{t,e})$.
\end{lemma}

\noindent{\em Proof:} See Appendix A.

\begin{lemma}
Under Assumptions 2 and 4, there exists a constant $C > 0$ that holds
\begin{equation}
\begin{aligned}
\frac{1}{I}\sum_{i=1}^I 
\mathbb{E}\!\left[
\bigl\| \mathbf{w}_i^{t,E} - \mathbf{w}^t \bigr\|^2
\right]
\le
C\,\eta_\mathbf{w}^2 E^2 \left( \sigma^2 + \gamma^2 \right).
\end{aligned}
\end{equation}
\end{lemma}

\noindent{\em Proof:} See Appendix B.


\begin{theorem}
Under Assumptions~1-4, suppose the step sizes satisfy
$\eta_\mathbf{w} \le \frac{c}{LE}$ and $\eta_\mathbf{v} \le \frac{c}{L}$ for a constant $c\in(0,1)$.
Define the mixed stationarity measure as 
\begin{equation}
\begin{aligned}
\mathcal{G}^{t,e}
\triangleq
\left\|\nabla_\mathbf{w} F\bigl(\mathbf{w}^t,\{\mathbf{v}^t\}\bigr)\right\|^2
+
\frac{1}{I}\sum_{i=1}^I
\left\|\nabla_{\mathbf{v}_i} f_i\bigl(x_i^{t,e}\bigr)\right\|^2.
\end{aligned}
\end{equation}
Then, after $T$ communication rounds with $E$ local steps per round, we have
\begin{equation}
\begin{aligned}
\frac{1}{TE}\sum_{t=0}^{T-1}\sum_{e=0}^{E-1}\mathbb{E}\!\left[\mathcal{G}^{t,e}\right]
\le &
\frac{4\big(F^{0}-F^{\star}\big)}{ET}\left(\frac{1}{\eta_\mathbf{w}}+\frac{1}{\eta_\mathbf{v}}\right)
\\&+ C_{1}L(\eta_\mathbf{w}+\eta_\mathbf{v})\sigma^{2}
+ C_{2}L\eta_\mathbf{w}E^{2}\gamma^{2} \\&+
C_3L(1-\omega)\eta_wE^2(\sigma^2+\gamma^2),
\end{aligned}
\end{equation}
where $F^{0}=F\!\left(\mathbf{w}^{0},\mathbf{v}^{0}\right)$ and $C_{1}$, $C_{2}$, $C_{3}$ are positive constants.
\end{theorem}

\noindent{\em Proof:} See Appendix C.

\subsubsection{Compressed Shared Backbone Transmission}
To further improve the shared backbone transmission efficiency in the UL and DL, we propose to optimize the UL traffic splitting ratio, the transmit power of each vehicle, and the DL traffic splitting ratio of its associated MBSs\footnote{The MBS DL transmit power is fixed and is included in the DT-based DL rate prediction; it is therefore not treated as an optimization variable.} in each communication round $t$ during the $k$th DT update. Let $\Xi^{\mathrm{DL}(k)}(t)$ denote the number of DL bits required to deliver the global shared backbone's parameters from the edge server to each vehicle, and $\rho_{b,i}^{\mathrm{DL}(k)}(t)$ denote the traffic splitting ratio from MBS $b$ to vehicle $i$. Let $\Xi_i^{\mathrm{UL}(k)}(t)$ denote the number of UL bits required for vehicle $i$ to upload its compressed shared backbone's update, $\rho_{i,b}^{\mathrm{UL}(k)}(t)$ denote the UL traffic splitting ratio from vehicle $i$ to MBS $b\in\mathcal B_i^{(k)}(t)$, and $P_{i,b}^{\mathrm{UL}(k)}(t)$ denote the corresponding UL transmit power. It is worth noting that $\Xi^{\mathrm{DL}(k)}(t)$ is identical for all vehicles,
since the server broadcasts the same global parameters. In contrast, $\Xi_i^{\mathrm{UL}(k)}(t)$ is vehicle-specific, because each vehicle may generate a different compressed local update depending on the local measured data. Then the UL SINR from vehicle $i$ to MBS $b$ is
\begin{equation}
\gamma_{i,b}^{\mathrm{UL}(k)}(t)
=
\frac{
P_{i,b}^{\mathrm{UL}(k)}(t)h_{i,b}^{\mathrm{UL}(k)}(t)
}
{
\sum\limits_{i'\neq i}
P_{i',b}^{\mathrm{UL}(k)}(t)h_{i',b}^{\mathrm{UL}(k)}(t)
+
\sigma^2
}, 
\end{equation}
where the channel gain $h_{i,b}^{\mathrm{UL}(k)}(t)$ is assumed to be constant during the communication round $t$, and $\sigma$ is the noise power. Thus the achievable UL data rate is
\begin{equation}
R_{i,b}^{\mathrm{UL}(k)}(t)
=
B_{i,b}^{\mathrm{UL}(k)}(t)
\log_2\left(
1+\gamma_{i,b}^{\mathrm{UL}(k)}(t)
\right).
\end{equation}
The UL transmission latency over the UL path includes the UL radio link latency between vehicle $i$ and MBS $b$ and the UL backhaul latency between MBS $b$ and the edge server, which is given by
\begin{equation}
T_{i,b}^{\mathrm{UL}(k)}(t)
=
\frac{
\rho_{i,b}^{\mathrm{UL}(k)}(t)\Xi_i^{\mathrm{UL}(k)}(t)
}
{
R_{i,b}^{\mathrm{UL}(k)}(t)
}+\frac{
\rho_{i,b}^{\mathrm{UL}(k)}(t)\Xi_i^{\mathrm{UL}(k)}(t)
}
{
c_{i,b}^{\mathrm{UL}(k)}(t)
}, 
\end{equation}
where $c_{i,b}^{\mathrm{UL}(k)}(t)$ is the UL Guaranteed Bit Rate (GBR)  between the MBS $b$ and the edge server. Similarly, the DL transmission latency over the DL path $(b,i)$ is
\begin{equation}
T_{b,i}^{\mathrm{DL}(k)}(t)
=
\frac{
\rho_{b,i}^{\mathrm{DL}(k)}(t)\Xi^{\mathrm{DL}(k)}(t)
}
{
R_{b,i}^{\mathrm{DL}(k)}(t)
}+\frac{
\rho_{b,i}^{\mathrm{DL}(k)}(t)\Xi^{\mathrm{DL}(k)}(t)
}
{
c_{b,i}^{\mathrm{DL}(k)}(t)
}, 
\end{equation}
where $R_{i,b}^{\mathrm{DL}(k)}(t)$ is the DL data rate from MBS $b$ to vehicle $i$ based on the predicted DT,  $c_{b,i}^{\mathrm{DL}(k)}(t)$ is the DL GBR between the edge server and the MBS $b$. Therefore, the joint UL and DL shared backbone transmission latency between the vehicle $i$ and the edge server via the MBS $b$ is 
\begin{equation}
L_{i,b}^{(k)}(t) = T_{i,b}^{\mathrm{UL}(k)}(t)
+
T_{b,i}^{\mathrm{DL}(k)}(t). 
\end{equation}
Then the corresponding transmission latency between the vehicle $i$ and the edge server is the maximum latency among all associated paths over $\mathcal B_i^{(k)}$, given by 
\begin{equation}
L_{i}^{(k)}(t) = \max_{b\in\mathcal B_i^{(k)}}
L_{i,b}^{(k)}(t). 
\end{equation}

Our objective is to minimize the joint UL and DL shared backbone transmission latency in each communication round $t$ during the $k$th DT update, which is formulated as 

\begin{algorithm}[!t]
\caption{{Efficient V2I DT Reconstruction}}
\label{alg:v2i_dt_update}
\KwIn{Vehicle set $\mathcal{I}$, V2I association sets $\{\mathcal{B}_i^{(k)}(t)\}_{i\in\mathcal{I}}$, communication rounds $T$, local epochs $E$}
\KwOut{Trained global shared backbone $\mathbf{w}^{(k)}$ and personalized heads $\{\mathbf{v}_i(k)\}_{i\in\mathcal{I}}$}

Initialize global shared backbone $\mathbf{w}$ and personalized heads $\{\mathbf v_i\}_{i\in\mathcal{I}}$\;

\For{Each DT reconstruction $k = 0,1,\dots,K-1$}{
    Build V2I association sets $\{\mathcal{B}_i^{(k)}\}_{i\in\mathcal{I}}$ using Algorithm~\ref{alg:decision}\;
    
    \For{Each communication round $t = 0,1,\dots,T-1$}{
        
        \For{each vehicle $i \in \mathcal{I}$ in parallel}{
            Set local shared backbone $\mathbf{w}_i^{t,0} \leftarrow \mathbf{w}^t$ and personalized head $\mathbf{v}_i^{t,0} \leftarrow \mathbf{v}_i^t$\;
            
            \For{Each local epoch $e = 0,1,\dots,E-1$}{
                Update $\mathbf{w}_i^{t,e+1}$ and $\mathbf{v}_i^{t,e+1}$ by minimizing the local multi-task loss $\mathcal{L}_i$ in Eq.~\eqref{eq:pfl_multitask_loss}\;
            }
            
            Set $\mathbf{w}_i^{t+1} \leftarrow \mathbf{w}_i^{t,E}$ and  $\mathbf{v}_i^{t+1} \leftarrow \mathbf{v}_i^{t,E}$\;
            
            Compute the local shared backbone update in Eq.~\eqref{eq:update}\;
            
            Compress the shared backbone update in Eq.~\eqref{eq:delta}\;
            
            Split and upload $\widetilde{\Delta}_i^{(k)}(t)$ through associated MBSs in terms of
            $\{\rho_{i,b}^{\mathrm{UL}(k)}(t), P_{i,b}^{\mathrm{UL}(k)}(t)\}$ in Problem 1\;
        }
        
        Aggregate the uploaded $\widetilde{\Delta}_i^{(k)}(t)$ in Eq.~\eqref{eq:compressw}\;
        
        Split and download $\mathbf{w}^{t+1}$ through $\mathcal{B}_i^{(k)}(t)$ in terms of
        $\rho_{b,i}^{\mathrm{DL}^{(k)}}(t)$ in Problem 1\;
    }
}

\textbf{return} $\mathbf{w}^{(k)}$, $\{\mathbf{v}_i^{(k)}\}_{i\in\mathcal{I}}$

\end{algorithm}

\begin{Prob}[Stage-I DT-Update Model-Exchange Latency Optimization]
\begingroup
\begin{subequations}
\label{eq:stage1_backbone_latency}
\begin{align}
& \hspace{-3.0em} \underset{
    \rho_{i,b}^{\mathrm{UL}(k)}(t),\,
    P_{i,b}^{\mathrm{UL}(k)}(t),\,
    \rho_{b,i}^{\mathrm{DL}(k)}(t)
  }{\min}
  \max_{i\in\mathcal I}
  L_i^{(k)}(t)
  \tag{\theparentequation}
  \label{eq:stage1_backbone_latency-obj}
  \\[2pt]
\mathrm{s.t.}\
& \mathrm{C1}:~0 \leq \rho_{i,b}^{\mathrm{UL}(k)}(t) \leq 1,
  \quad
  \forall i\in\mathcal I,\ b\in\mathcal B_i(t),
  \label{eq:stage1_backbone_latency-c1}
  \\
& \mathrm{C2}:~\sum_{b\in\mathcal B_i(t)}
  \rho_{i,b}^{\mathrm{UL}(k)}(t) = 1,
  \quad
  \forall i\in\mathcal I,
  \label{eq:stage1_backbone_latency-c2}
  \\
& \mathrm{C3}:~0 \leq P_{i,b}^{\mathrm{UL}(k)}(t) \leq P_i^{\max},
  \quad
  \forall i\in\mathcal I,\ b\in\mathcal B_i^{(k)}(t),
  \label{eq:stage1_backbone_latency-c3}
  \\
& \mathrm{C4}:~\sum_{b\in\mathcal B_i^{(k)}(t)}
  P_{i,b}^{\mathrm{UL}(k)}(t) \leq P_i^{\max},
  \quad
  \forall i\in\mathcal I,
  \label{eq:stage1_backbone_latency-c4}
  \\
& \mathrm{C5}:~0 \leq \rho_{b,i}^{\mathrm{DL}(k)}(t) \leq 1,
  \quad
  \forall i\in\mathcal I,\ b\in\mathcal B_i^{(k)}(t),
  \label{eq:stage1_backbone_latency-c5}
  \\
& \mathrm{C6}:~\sum_{b\in\mathcal B_i^{(k)}(t)}
  \rho_{b,i}^{\mathrm{DL}(k)}(t) = 1,
  \quad
  \forall i\in\mathcal I,
  \label{eq:stage1_backbone_latency-c6}
  \\
& \mathrm{C7}:~L_i^{(k)}(t) \leq L_{CSB}^{\max},
  \quad
  \forall i\in\mathcal I,
  \label{eq:stage1_backbone_latency-c7}
\end{align}
\end{subequations}
\endgroup

\end{Prob}
\noindent where constraints C1-C2 and C5-C6 ensure that each vehicle splits its UL and DL shared backbone's update among its associated MBSs. C3-C4 impose the UL transmit-power budget of each vehicle. C7 defines the shared backbone's transmission latency budget for a communication round. Algorithm~\ref{alg:v2i_dt_update} summarizes the proposed V2I DT reconstruction method.

In the following, we formulate the second problem to form the closed-loop perspective.



\subsection{E2E Multi-Service of Traffic Transmission}

Once the DT is updated, each vehicle $i$ further transmits $\mathcal{Q}=\{1,\ldots,Q\}$ types of traffic\footnote{According to the 3GPP standard \cite{Garcia2021NRV2X,3gpp.38.885}, the V2X traffic service types includes broadcast, unicast and groupcast. } to the edge server and the traffic of each type is split via the associated MBS set $B_i^{(k)}$. To minimize the service-level E2E UL and DL latency in the whole system where each vehicle $i$ optimizes the UL traffic splitting ratio and the UL transmit power over the associated MBSs, while the edge server optimizes the DL traffic splitting ratio\footnote{The DL transmit power of the MBS is determined by the reconstructed DT, which is thereby not optimized here.}. We consider block-fading channels where the UL channel $h_{i,b}^{\mathrm{UL}(k)}(m)$ and the DL channel $h_{b,i}^{\mathrm{DL}(k)}(m)$ between vehicle $i$ and MBS $b$ are assumed to be constant within a block and vary independently across $\mathcal{M} = \{1,\ldots,M\} $ blocks in the $k$th DT update. Accordingly, joint decisions are made at the beginning of each block and kept fixed for the whole block.

In each block $m \in \mathcal{M}$, suppose the number of arrived packets of traffic type $q$ at the vehicle $i$ in block $m$ is a random variable $\lambda_{i,q}^{(k)}(m)$ with mean
$\mathbb{E}\!\left[\lambda_{i,q}^{(k)}(m)\right]=\bar{\lambda}_{i,q}^{(k)}$.
Each packet of traffic type $q$ has a fixed size of $M_q$ bits, and all packets with the same type share the same size.
Meanwhile, let $W_{i,q}^{(k)}(m)$ denote the number of queueing packets of type $q$ at vehicle $i$ at the beginning of block $m$ with the mean
$\mathbb{E}\!\left[W_{i,q}^{(k)}(m)\right]=\bar{W}_{i,q}$. Let $P_{i,b}^{\mathrm{UL}(k)}(m)$ denote the corresponding UL transmit power. The UL SINR from vehicle $i$ to MBS $b$ is
\begin{equation}
\gamma_{i,b}^{\mathrm{UL}(k)}(m)
=
\frac{
P_{i,b}^{\mathrm{UL}(k)}(m)h_{i,b}^{\mathrm{UL}(k)}(m)
}
{
\sum\limits_{i'\neq i}
P_{i',b}^{\mathrm{UL}(k)}(m)h_{i',b}^{\mathrm{UL}(k)}(m)
+
\sigma^2
}, 
\end{equation}
and the corresponding achievable UL data rate is
\begin{equation}
R_{i,b}^{\mathrm{UL}(k)}(m)
=
B_{i,b}^{\mathrm{UL}(k)}(m)
\log_2\left(
1+\gamma_{i,b}^{\mathrm{UL}(k)}(m)
\right).
\end{equation}
We denote the UL traffic splitting ratio of type $q$ traffic from vehicle $i$ to MBS $b$ as $\alpha_{i,b,q}^{\mathrm{UL}(k)}(m), \forall i\in\mathcal{I},\ \forall q\in\mathcal{Q}$. Therefore, the UL transmission latency for traffic type $q$,  including the radio link latency from vehicle $i$ to MBS $b$ and the backhaul latency from MBS $b$ to the edge server is
\begin{equation}
\begin{split}
T^{\mathrm{UL}(k)}_{i,b,q}(m)&=
\frac{\alpha_{i,b,q}^{\mathrm{UL}(k)}(m)\big(\lambda_{i,q}^{\mathrm{UL}(k)}(m)+W^{\mathrm{UL}(k)}_{i,q}(m)\big)M_q}{R_{i,b}^{\mathrm{UL}(k)}(m)} \\&+
\frac{\alpha_{i,b,q}^{\mathrm{UL}(k)}(m)\big(\lambda_{i,q}^{\mathrm{UL}(k)}(m)+W^{\mathrm{UL}(k)}_{i,q}(m)\big)M_q}{c_{i,b,q}^{\mathrm{UL}(k)}(m)},
\end{split}
\end{equation}
where $c_{i,b,q}^{\mathrm{UL}}(m)$ follows the distribution of the corresponding
Quality-of-Service (QoS) flow characteristics of service type $q$. Similarly, the DL transmission latency for traffic type $q$ over the same path is 
\begin{equation}
\begin{split}
T^{\mathrm{DL}(k)}_{b,i,q}(m)&=
\frac{\alpha_{b,i,q}^{\mathrm{DL}(k)}(m)\big(\lambda_{i,q}^{\mathrm{DL}(k)}(m)+W^{\mathrm{DL}(k)}_{i,q}(m)\big)M_q}{R_{b,i}^{\mathrm{DL}(k)}(m)} \\&+
\frac{\alpha_{b,i,q}^{\mathrm{DL}(k)}(m)\big(\lambda_{i,q}^{\mathrm{DL}(k)}(m)+W^{\mathrm{DL}}_{i,q}(m)\big)M_q}{c^{\mathrm{DL}(k)}_{b,i,q}(m)}. 
\end{split}
\end{equation}
Thus, the joint E2E UL and DL latency for traffic type $q$ between vehicle $i$ and the edge server via MBS $b$ during the block $m$ is 
\begin{equation}
L_{i,b,q}^{(k)}(m) = T^{\mathrm{UL}(k)}_{i,b,q}(m) + T^{\mathrm{DL}(k)}_{b,i,q}(m). 
\end{equation}
Then the corresponding transmission latency for traffic type $q$ between vehicle $i$ and the edge server is the maximum latency among all associated paths over $\mathcal B_i^{(k)}$, given by 
\begin{equation}
L_{i,q}^{(k)}(m) = \max_{b\in\mathcal{B}_i^{(k)}} \ L_{i,b,q}^{(k)}(m).
\end{equation}
Our objective is to minimize the joint E2E UL and DL latency per block of all vehicles with all traffic types. Accordingly, we formulate the second optimization problem as follows:
\begin{Prob}[Stage-II Service-Aware Communication Latency Optimization]
\label{prob:P1}
\begingroup
\begin{subequations}
\label{eq:prob1}
\begin{align}
& \hspace{-2.0em} \underset{
    \alpha_{i,b,q}^{\mathrm{UL}(k)}(m),\,
    P_{i,b}^{\mathrm{UL}(k)}(m),\,
    \alpha_{i,b,q}^{\mathrm{DL}(k)}(m)
  }{\min}
  \sum_{i\in\mathcal{I}}\sum_{q\in\mathcal{Q}}
  L_{i,q}^{(k)}(m)
  \tag{\theparentequation}
  \label{eq:prob1-obj}
  \\[2pt]
\mathrm{s.t.}\
& \mathrm{C1}:~0 \leq \alpha_{i,b,q}^{\mathrm{UL}(k)}(m) \leq 1,
  \quad
  \forall i\in\mathcal{I},\ b\in\mathcal{B}_i^{(k)},\
  q\in\mathcal{Q},
  \label{eq:prob1-c1}
  \\
& \mathrm{C2}:~\sum_{b\in\mathcal{B}_i^{(k)}}
  \alpha_{i,b,q}^{\mathrm{UL}(k)}(m) = 1,
  \quad
  \forall i\in\mathcal{I},\ q\in\mathcal{Q},
  \label{eq:prob1-c2}
  \\
& \mathrm{C3}:~0 \leq P_{i,b}^{\mathrm{UL}(k)}(m) \leq P_i^{\max},
  \quad
  \forall i\in\mathcal{I},\ b\in\mathcal{B}_i^{(k)},
  \label{eq:prob1-c3}
  \\
& \mathrm{C4}:~\sum_{b\in\mathcal{B}_i^{(k)}}
  P_{i,b}^{\mathrm{UL}(k)}(m) \leq P_i^{\max},
  \quad
  \forall i\in\mathcal{I},
  \label{eq:prob1-c4}
  \\
& \mathrm{C5}:~0 \leq \alpha_{b,i,q}^{\mathrm{DL}(k)}(m) \leq 1,
  \quad
  \forall i\in\mathcal{I},\ b\in\mathcal{B}_i^{(k)},\
  q\in\mathcal{Q},
  \label{eq:prob1-c5}
  \\
& \mathrm{C6}:~\sum_{b\in\mathcal{B}_i^{(k)}}
  \alpha_{b,i,q}^{\mathrm{DL}(k)}(m) = 1,
  \quad
  \forall i\in\mathcal{I},\ q\in\mathcal{Q},
  \label{eq:prob1-c6}
  \\
& \mathrm{C7}:~\max_{b\in\mathcal{B}_i^{(k)}}
  L_{i,b,q}^{(k)}(m) \leq L_q^{\max},
  \quad
  \forall q\in\mathcal{Q},
  \label{eq:prob1-c7}
\end{align}
\end{subequations}
\endgroup
\end{Prob}
\noindent where $L_q^{\max}$ in C7 is the latency budget of type $q$.

\begin{remark}
Problems 1 and 2 correspond to two distinct but temporally coupled stages within each DT-update cycle: Problem 1 is executed during the DT reconstruction to minimize the joint UL and DL latency of the global model exchange in each communication round; Problem 2 is  executed after DT reconstruction and minimizes the aggregated service-aware joint  UL and DL latency in each transmission block, subject to service-specific 
latency budgets. The two problems are operationally decoupled by their execution order and timescales and are solved sequentially. However, they share the same agent roles and structurally aligned action components: vehicle agents determine UL traffic splitting and UL power allocation, while the edge server's agent determines DL traffic splitting. Meanwhile, the multi-service types $q\in \mathcal{Q}$ introduced in Problem 2 corresponds to the global model as the single service type ($|\mathcal{Q}|$=1). Therefore, their action spaces are aligned but not identical, which enable cross-problem policy initialization without merging the two objectives.
\end{remark}



\section{Unified SAGE-MAPPO Framework for Closed-Loop Dual Optimization Problems}
\label{secV}

\begin{figure*}[t]
\centering
\includegraphics[width=0.95\textwidth]{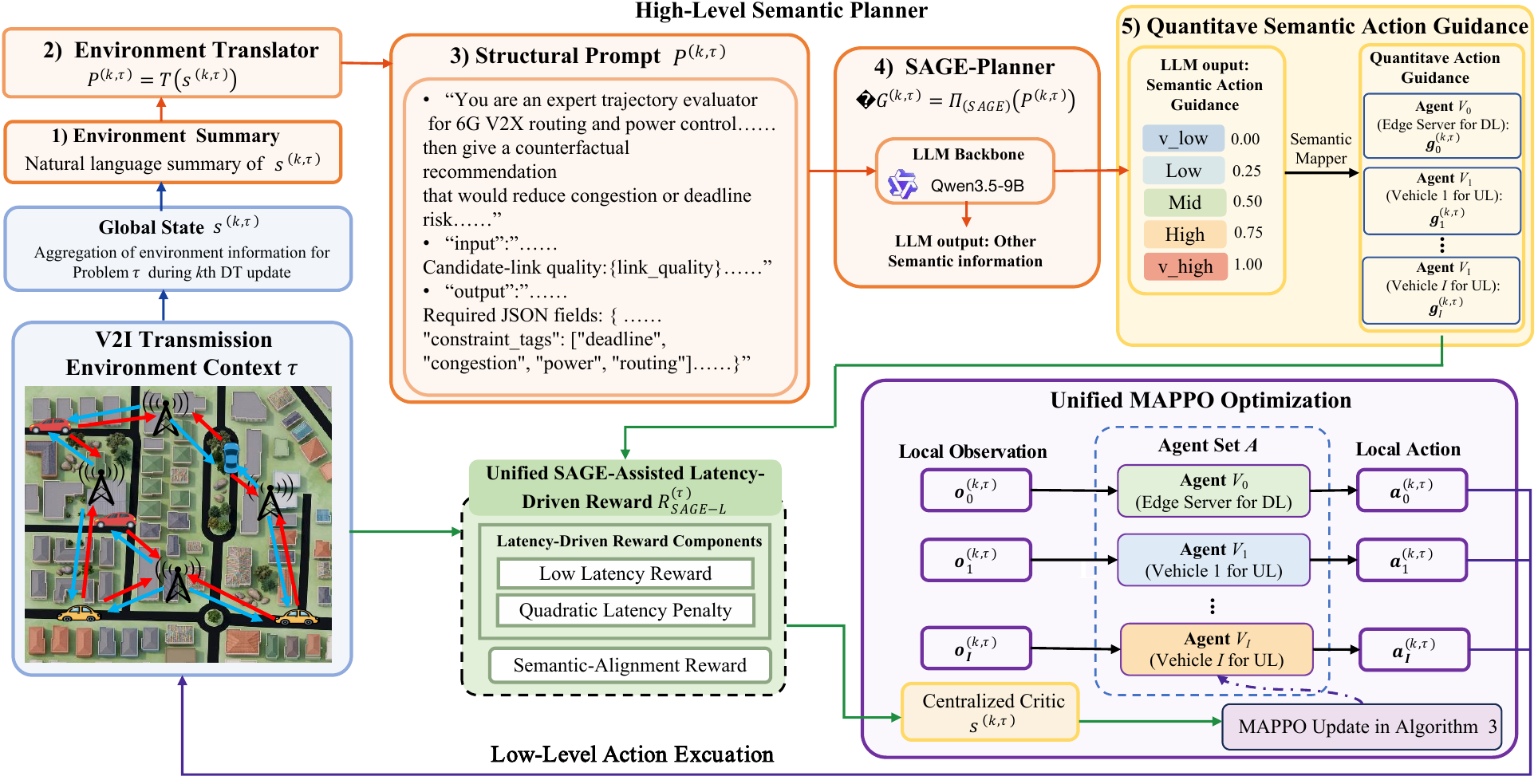} 
\vspace{-0.2cm}
\caption{Overall Framework of Unified SAGE-MAPPO.}
\vspace{-0.6cm}
\label{fig_fedcompre}
\end{figure*}

To solve Problems~1 and~2 in a unified manner, we employ the same LLM-guided MARL architecture, implemented by a SAGE-MAPPO framework. The key idea is that Problems~1 and~2 share the same agent roles and highly consistent control dimensions: the vehicle agents optimize UL traffic splitting and
transmit power allocation, while the edge server optimizes DL traffic splitting. Therefore, instead of designing two separate MARL solvers, we design
a unified LLM-MARL architecture and instantiate it with problem-dependent state
semantics and reward functions. In particular, Problem 1 corresponds to
the efficient shared backbone's transmission during DT reconstruction, whereas Problem 2 corresponds to the service-aware E2E joint UL and DL
latency optimization.

\subsubsection{Unified SAGE-MAPPO Components} The fundamental components of SAGE-MAPPO are described as follows:

\textbf{(a) Unified Agents}: We define a heterogeneous low-level multi-agent set consisting of $|\mathcal I|$ vehicle
agents for UL and an edge server controller for DL. Specifically, the UL agent set is
defined as $\mathcal A^{\mathrm{UL}}=\{V_i \mid i\in\mathcal I\}$,
where each vehicle agent $V_i$ determines the UL traffic splitting ratios
and transmit power allocation over its associated MBS set $\mathcal B_i^{(k)}$.
To remain consistent with Problems~1 and~2, the DL decisions are made by the edge server agent, denoted by $\mathcal A^{\mathrm{DL}}=\{V_0\}$,
which determines the DL traffic splitting ratios from the associated MBSs to the vehicles. Therefore, the overall agent set is $\mathcal A=\mathcal A^{\mathrm{UL}}\cup\mathcal A^{\mathrm{DL}} $. Meanwhile, a high-level planner $\Pi_{\mathrm{SAGE}}$ is further introduced for training-time semantic supervision.

\textbf{(b) Local Observations}: 
For the vehicle agent $V_i$, the local observation under Problem
$\tau\in\{1,2\}$ during the $k$th DT update is defined as 
\begin{equation}
\mathbf o_{i}^{(k,\tau)}=
\Big[
\{\gamma_{i,b}^{\mathrm{UL}}\}_{b\in\mathcal B_i^{(k)}},
\{\widehat{\gamma}_{b,i}^{\mathrm{DL}}\}_{b\in\mathcal B_i^{(k)}},
\mathbf z_i^{(k,\tau)},
P_i^{\max}
\Big]^T,
\end{equation}
where $\mathbf z_i^{(k, \tau)}$ denotes the problem-dependent latency state. For
Problem 1, $\mathbf z_i^{(k,1)}$ includes the compressed shared backbone transmission
states such as $\Xi_i^{\mathrm{UL}(k)}$ and $\Xi_i^{\mathrm{DL}(k)}$. For Problem
2, $\mathbf z_i^{(k,2)}$ includes the service-aware traffic states such
as $\{\lambda_{i,q}^{(k)},W_{i,q}^{(k)}\}_{q\in\mathcal Q}$ and the corresponding
path-latency statistics. For the DL controller $V_0$, the local observation is defined as 
\begin{equation}
\mathbf o_{0}^{(k,\tau)}=
\Big[
\{\widehat{\gamma}_{b,i}^{\mathrm{DL}}\}_{i\in\mathcal I,\,b\in\mathcal B_i^{(k)}},
\{\widehat{R}_{b,i}^{\mathrm{DL}}\}_{i\in\mathcal I,\,b\in\mathcal B_i^{(k)}},
\mathbf z_{0}^{(k,\tau)}
\Big]^T,
\end{equation}
where $\mathbf z_{0}^{(k,\tau)}$ collects problem-dependent global DL
states at $V_0$, i.e., DL shared-backbone sizes for Problem~1 or service-aware
DL traffic states for Problem~2. 

\textbf{(c) Global States}: 
During centralized training, the centralized critic observes the global state during the $k$th DT update as
\begin{equation}
\mathbf s^{(k,\tau)}=
\Big[
\{\mathbf o_i^{(k,\tau)}\}_{i\in\mathcal I},
\mathbf o_0^{(k,\tau)},
\mathbf u^{(k,\tau)}
\Big],
\end{equation}
where $\mathbf u^{(k,\tau)}(t)$ denotes the task-level latency summary. For Problem 1, it corresponds to the round-level joint UL and DL shared-backbone transmission
latency; for Problem 2, it corresponds to the service-aware joint UL and DL latency statistics under the $k$th reconstructed DT.

{\setlength{\abovedisplayskip}{4pt}
\setlength{\belowdisplayskip}{4pt}
\begingroup
\begin{figure*}[!b]
\vspace{-0.2cm}
\hrulefill
\begin{equation}
\small
\begin{split}
&R_{\text{SAGE-L}}^{(k,1)}(t)
={}
\underbrace{\phi_L \cdot \exp\left(-\alpha_L \cdot  \frac{\max_{i\in\mathcal I}\
L_{i}^{(k,1)}(t)}{L^{\max(1)}}\right)}_{\text{Low Latency Reward}} - \underbrace{\lambda_{\mathrm{vio}} \cdot \frac{1}{|\mathcal I|} \sum_{i\in\mathcal I}\mathbf{1}_{\{L_{i}^{(k,1)}(t) > L^{\max(1)}\}} \cdot \left(\frac{L_{i}^{(k,1)}(t)-L^{\max(1)}}{L^{\max(1)}}\right)^2}_{\text{High Latency Penalty}} + \\
& \underbrace{\lambda_{\mathrm{sem}}\frac{1}{|\mathcal I|} \sum_{i\in\mathcal I}\sum_{b\in B_i^{(k)} }\left\{ \exp\left(-  |\rho_{i,b}^{\mathrm{UL-pwr}}(t)  - \hat{\rho}_{i,b}^{\mathrm{UL-pwr}}(t)|\right) +\exp\left(-|\rho_{i,b}^{\mathrm{UL}}(t) - \hat{\rho}_{i,b}^{\mathrm{UL-split}}(t)|\right)+\exp\left(-|\rho_{b,i}^{\mathrm{DL}}(t) - \hat{\rho}_{b,i}^{\mathrm{DL-split}}(t)|\right)\right\}}_{\text{Semantic-Alignment Reward}}.
\end{split}
\label{eq:sage_reward1}
\end{equation}
\end{figure*}
\endgroup}

{\setlength{\abovedisplayskip}{4pt}
\setlength{\belowdisplayskip}{4pt}
\begingroup
\footnotesize
\begin{figure*}[!b]
\vspace{-0.7cm}
\begin{equation}
\footnotesize
\begin{split}
& R_{\text{SAGE-L}}^{(k,2)}(m)
={}
\underbrace{\phi_L \cdot \exp\left(-\alpha_L \cdot \frac{1}{|\mathcal I|}\sum_{i\in\mathcal I}\sum_{q\in\mathcal Q} \frac{L_{i,q}^{(k,2)}(m)}{L^{max (2)}_q}\right)}_{\text{Low Latency Reward}} - \underbrace{\lambda_{\mathrm{vio}} \cdot \frac{1}{|\mathcal I|} \sum_{i\in\mathcal I} \sum_{q\in\mathcal Q}\mathbf{1}_{ \left\{L_{i,q}^{(k,2)}(m) > L^{max (2)}_q \right\}} \cdot \left(\frac{L_{i,q}^{(k,2)}(m)-L^{max (2)}_q}{L^{max (2)}_q}\right)^2}_{\text{High Latency Penalty}} + \\
& \underbrace{\lambda_{\mathrm{sem}}\frac{1}{|\mathcal I|} \sum_{i\in\mathcal I}\sum_{b\in B_i^{(k)} }\left\{\exp\left(- |\rho_{i,b}^{\mathrm{UL-pwr}}(m)  - \hat{\rho}_{i,b}^{\mathrm{UL-pwr}}(m)|\right) +\sum_{q\in\mathcal Q}\left[\exp\left( - |\alpha_{i,b,q}^{\mathrm{UL}}(m) - \hat{\alpha}_{i,b,q}^{\mathrm{UL-split}}(m)|\right)+\exp\left(- |\alpha_{b,i,q}^{\mathrm{DL}}(m) - \hat{\alpha}_{b,i,q}^{\mathrm{DL-split}}(m)|\right)\right]\right\}}_{\text{Semantic-Alignment Reward}}.
\end{split}
\label{eq:sage_reward2}
\end{equation}
\end{figure*}
\endgroup}

\textbf{(d) SAGE Planner and Environment Translator}: To integrate semantic reasoning into the unified MARL framework, the global
state is first translated into a compact prompt $\mathbf P^{(k,\tau)}=\mathcal T\!\left(\mathbf s^{(k,\tau)}\right)$,
where $\mathcal T(\cdot)$ denotes the environment translator. The SAGE planner
then outputs the semantic guidance $\mathbf G^{(k,\tau)}=\Pi_{\mathrm{SAGE}}\!\left(\mathbf P^{(k,\tau)}\right)$. For Problem~1, the direct prompted action guidance for vehicle agent $V_i$ is defined as
\begin{equation}
\mathbf g_i^{(k,1)}=
\left\{
\hat{\rho}_{i,b}^{\mathrm{UL-split}},\,
\hat{\rho}_{i,b}^{\mathrm{UL-pwr}}
\right\}^T_{b\in\mathcal B_i^{(k)}}, 
\end{equation}
where $\hat{\rho}_{i,b}^{\mathrm{UL-split}} \in [0,1]$ and $\hat{\rho}_{i,b}^{\mathrm{UL-pwr}} = \hat{P}^{\mathrm{UL}}_{i,b}/P_i^{\max}\in [0,1]$ are the guided UL traffic splitting tendency and UL power allocation tendency, respectively. The direct prompt action guidance for the edge server 
$V_0$ is 
\begin{equation}
\mathbf g_0^{(k,1)}=
\left\{
\hat{\rho}_{b,i}^{\mathrm{DL-split}}
\right\}_{i\in\mathcal I,\,b\in\mathcal B_i^{(k)}}.
\end{equation}
For Problem~2, the guidance is extended to the service-aware case. The vehicle
agent $V_i$ receives
\begin{equation}
\mathbf g_i^{(k,2)}
=
\left\{
\hat{\alpha}_{i,b,q}^{\mathrm{UL-split}},
\hat{\rho}_{i,b}^{\mathrm{UL-pwr}}
\right\}_{b\in\mathcal B_i^{(k)},\,q\in\mathcal Q},
\label{eq:sage_guidance_vehicle_p2}
\end{equation}
where the UL traffic splitting guidance is service-dependent, while the UL
power guidance remains link-dependent. The edge-server guidance is
\begin{equation}
\mathbf g_0^{(k,2)}
=
\left\{
\hat{\alpha}_{b,i,q}^{\mathrm{DL-split}}
\right\}_{i\in\mathcal I,\,b\in\mathcal B_i^{(k)},\,q\in\mathcal Q}.
\label{eq:sage_guidance_server_p2}
\end{equation}


\textbf{(e) Actions}: To keep the unified SAGE-MAPPO architecture consistent with Problems~1 and~2, we define problem-dependent action spaces with the same agent roles. For
Problem~1, each vehicle agent $V_i$ controls the UL shared backbone's
transmission over its associated MBSs:
\begin{equation}
\mathbf a_i^{(k,1)}(t)
=
\left\{
\rho_{i,b}^{\mathrm{UL}(k)}(t),
P_{i,b}^{\mathrm{UL}(k)}(t)
\right\}_{b\in\mathcal B_i^{(k)}} ,
\label{eq:vehicle_action_p1}
\end{equation}
where $\rho_{i,b}^{\mathrm{UL}(k)}(t)$ is the UL traffic splitting ratio and
$P_{i,b}^{\mathrm{UL}(k)}(t)$ is the UL transmit power. The edge-server agent
controls the DL shared backbone's transmission:
\begin{equation}
\mathbf a_0^{(k,1)}(t)
=
\left\{
\rho_{b,i}^{\mathrm{DL}(k)}(t)
\right\}_{i\in\mathcal I,\,b\in\mathcal B_i^{(k)}}.
\label{eq:server_action_p1}
\end{equation}

For Problem~2, the traffic splitting action is further extended to each
service type $q\in\mathcal Q$. The vehicle-agent action becomes
\begin{equation}
\mathbf a_i^{(k,2)}(m)
=
\left\{
\alpha_{i,b,q}^{\mathrm{UL}(k)}(m)
\right\}_{b\in\mathcal B_i^{(k)},\,q\in\mathcal Q}
\cup
\left\{
P_{i,b}^{\mathrm{UL}(k)}(m)
\right\}_{b\in\mathcal B_i^{(k)}} ,
\label{eq:vehicle_action_p2}
\end{equation}
where $\alpha_{i,b,q}^{\mathrm{UL}(k)}(m)$ is the service-aware UL
traffic splitting ratio and $P_{i,b}^{\mathrm{UL}(k)}(m)$ is the link-level UL
transmit power. The edge-server action is
\begin{equation}
\mathbf a_0^{(k,2)}(m)
=
\left\{
\alpha_{b,i,q}^{\mathrm{DL}(k)}(m)
\right\}_{i\in\mathcal I,\,b\in\mathcal B_i^{(k)},\,q\in\mathcal Q}.
\label{eq:server_action_p2}
\end{equation}



\textbf{(f) Unified Latency-driven Reward:} 
We propose a unified reward function by maintaining a single SAGE-MAPPO architecture across the two optimization problems: 
\begin{equation}
R_{\text{SAGE-L}}^{(k,\tau)}
=
\phi_L\Phi^{(k,\tau)}
-\lambda_{\mathrm{vio}}\Psi^{(k,\tau)}
+\lambda_{\mathrm{sem}}\Omega^{(k,\tau)},
\label{eq:unified_reward_template}
\end{equation}
where the designed reward includes three components. The first term $\Phi^{(k,\tau)}$ is the exponential latency reward that encourages to minimize the E2E latency in Problems 1 or 2, while $\phi_{L}>0$ is the latency reward weight. The second term $\Psi^{(k,\tau)}$ is the high latency penalty using the quadratic function that is against the actions whose resulting latency exceeds the latency budget for both problems, while $\lambda_{\mathrm{vio}}>0$ is the latency-violation penalty factor. The third term $\Omega_{k,\tau}$ is the semantic-alignment reward that measures the average consistency between the SAGE guidance and the actual low-level execution across all agents. The dynamic alignment value $\lambda_{\mathrm{sem}}\geq 0$ becomes insignificant when the learned policies become more consistent with the planned strategy. As a result, the reward $R_{\text{SAGE-L}}^{(k,1)}(t)$ for Problem~1 is instantiated as Eq. (\ref{eq:sage_reward1}), where the corresponding SAGE action guidance and the actual low-level action execution are defined in Eq. (43)-(44), and Eq. (47)-(48), respectively. Similarly, for Problem~2, the reward $R_{\text{SAGE-L}}^{(k,2)}(m)$ is instantiated as Eq. (\ref{eq:sage_reward2}), where the corresponding SAGE action guidance and the actual low-level action execution are defined in Eq. (45)-(46) and Eq. (49)-(50), respectively. In addition, $\alpha_{L}>0$ in Eq. (\ref{eq:sage_reward1}) and Eq. (\ref{eq:sage_reward2}) is the scaling parameter.

\begin{remark}
When $\lambda_{\mathrm{sem}}=0$, the reward in Eq.~\eqref{eq:unified_reward_template} reduces to a purely latency-driven reward with a latency-violation penalty, and the proposed SAGE-MAPPO framework generalizes a standard latency-oriented MAPPO design as a special case.
\end{remark}


\subsubsection{Training Procedure of Unified SAGE-MAPPO}
To handle partial observability and temporal dynamics, each agent employs a lightweight encoder and a recurrent policy network. For agent $a\in\mathcal A$,
the local observation is first mapped to a compact feature $\mathbf e_a^{(k,\tau)}=
f_{\mathrm{enc}}\!(\mathbf o_a^{(k,\tau)};\boldsymbol{\vartheta}_{\mathrm{enc}})$, and then fed into a Gated Recurrent Unit (GRU) to update the hidden state $\mathbf h_a^{(k,\tau)}=
\mathrm{GRU}\!(\mathbf h_a^{(k,\tau)},\mathbf e_a^{(k,\tau)})$.
The decision feature is thus $\mathbf z_a^{(k,\tau)}=
[
\mathbf e_a^{(k,\tau)},\mathbf h_a^{(k,\tau)}]$. For each vehicle agent, the actor network is factorized into an UL traffic splitting head and an UL power control head:
$\pi_i^{(k,\tau)}\!(
\mathbf a_i^{(k,\tau)}\mid \mathbf z_i^{(k,\tau)})
=
\pi_{i}^{\mathrm{UL-split}}
\cdot
\pi_{i}^{\mathrm{UL-pwr}}$,
while the edge server controller adopts a traffic splitting policy $\pi_{0}^{(k,\tau)}\!(
\mathbf a_0^{(k,\tau)}\mid \mathbf z_0^{(k,\tau)}
)
=
\pi_{0}^{\mathrm{DL-split}}$. At each decision step, the centralized state is translated into a semantic prompt, and the SAGE planner outputs the guidance $\mathbf G^{(k,\tau)}$. The agents then sample their actions from the current
policies, interact with the environment, and receive the problem-dependent reward
defined in Eq. \eqref{eq:sage_reward1}-\eqref{eq:sage_reward2}. The rollout trajectories are stored in an on-policy buffer and used to update the actor and critic networks via the
standard clipped MAPPO objective.  In Algorithm~\ref{alg:sage_mappo_algorithm}, the generic step index $n$ corresponds to the communication round $t$ for Problem~1 and the transmission block $m$ for Problem~2.

\begin{algorithm}[t]
\caption{{Unified SAGE-MAPPO for Problems 1 and 2}}
\label{alg:sage_mappo_algorithm}

\KwIn{
DT update set $\mathcal K=\{0,\ldots,K-1\}$, environment context
$\tau\in\{1,2\}$, vehicle agent set
$\mathcal A^{\mathrm{UL}}=\{V_i|i\in\mathcal I\}$, edge server controller
$V_0$, training episodes $M$, decision steps $T_{\tau}$, rollout length
$L_{\mathrm{roll}}$, PPO update epochs $K_{\mathrm{PPO}}$
}

\KwOut{
Trained unified policies
$\{\pi_i^{(k,\tau)}\}_{i\in\mathcal I}$ and $\pi_0^{(k,\tau)}$
}

Initialize vehicle agent actors $\{\pi_i\}_{i\in\mathcal I}$, edge-server actor
$\pi_0$, centralized critic, SAGE planner $\Pi_{\mathrm{SAGE}}$, translator
$\mathcal T(\cdot)$, and rollout buffer $\mathcal D\leftarrow\emptyset$\;

\For{each DT update $k=1,\ldots,K-1$}{ 

    \For{problem $\tau=1,2$}{
        \If{$\tau=2$}{
            Warm-start the Problem~2 policies using the policy parameters learned from Problem~1 under the current DT update $k$\;
        }

        \For{episode $e=1,\ldots,M$}{

            Reset the problem-$\tau$ environment\;
            Obtain initial state $\mathbf s^{(k,\tau)}(0)$ and observations
            $\{\mathbf o_i^{(k,\tau)}(0)\}_{i\in\mathcal I}$,
            $\mathbf o_0^{(k,\tau)}(0)$\;

            \For{step $n=0,\ldots,T_{\tau}-1$}{

                Generate semantic guidance
                $\mathbf G^{(k,\tau)}(n)
                =
                \Pi_{\mathrm{SAGE}}
                \big(\mathcal T(\mathbf s^{(k,\tau)}(n))\big)$\;

                \For{each vehicle agent $V_i\in\mathcal A^{\mathrm{UL}}$}{
                    Encode $\mathbf o_i^{(k,\tau)}(n)$ and sample
                    $\mathbf a_i^{(k,\tau)}(n)$ from $\pi_i^{(k,\tau)}$ \;
                }

                Encode $\mathbf o_0^{(k,\tau)}(n)$ and sample
                $\mathbf a_0^{(k,\tau)}(n)$ from $\pi_0^{(k,\tau)}$\;

                Project all sampled actions onto the feasible region of
                Problem~$\tau$\;


                Observe $\mathbf s^{(k,\tau)}(n+1)$ and compute
                $R_{\mathrm{SAGE\mbox{-}L}}^{(k,\tau)}(n)$ using
                Eq.~\eqref{eq:sage_reward1} if $\tau=1$, or
                Eq.~\eqref{eq:sage_reward2} if $\tau=2$\;

                Store the transition in $\mathcal D$\;

                \If{$|\mathcal D|\ge L_{\mathrm{roll}}$}{
                    Update all actors and
                    centralized critic using the clipped MAPPO objective with
                    $K_{\mathrm{PPO}}$ epochs\;
                    Clear $\mathcal D$\;
                }
            }
        }

        Store the trained policies
        $\{\pi_i^{(k,\tau)}\}_{i\in\mathcal I}$ and $\pi_0^{(k,\tau)}$\;
    }
}

\Return{
$\{\pi_i^{(k,\tau)}\}_{i\in\mathcal I,k\in\mathcal K,\tau\in\{1,2\}}$,
$\{\pi_0^{(k,\tau)}\}_{k\in\mathcal K,\tau\in\{1,2\}}$
}

\end{algorithm}

\begin{table}[!t]
\centering
\caption{CPFL Hyperparameter Settings}
\vspace{-0.2cm}
\label{tab:pfl_params}
\resizebox{0.99\columnwidth}{!}{%
\begin{tabular}{p{3.2cm}p{5cm}}
\toprule
\textbf{Parameter} & \textbf{Configuration} \\
\midrule
Input feature dimension & 5 Env channels + relative distance features \\
Shared backbone & MLP 102$\to$256$\to$1024$\to$512, SiLU + LayerNorm \\
PL task head & Linear 512$\to K$, dropout = 0.25 \\
KPI task head & MLP 512$\to$384$\to$129, dropout = 0.12 \\
Initialization & Xavier uniform \\
Optimizer & AdamW \\
Learning rate (backbone/PL head/KPI head) & $2.6\times 10^{-4}$/$1.4\times 2.6\times 10^{-4}$/$1.7\times 2.6\times 10^{-4}$ \\
Batch size & $128$ \\
Loss function & Huber loss, $\delta \approx 0.57$ \\
Gradient clip & Norm threshold = $1.88$ \\
Communication rounds $T$ & $1200$ \\
Local epochs $E$& $3$ \\
Gradient compression & Top-37\% sparsification + 8-bit quantization \\
\bottomrule
\end{tabular}%
}\vspace{0.3cm}
\end{table}

\begin{table}[!t]
\centering
\caption{Unified SAGE-MAPPO Hyperparameter Settings}
\vspace{-0.2cm}
\label{tab:sage_params}
\resizebox{0.99\columnwidth}{!}{%
\begin{tabular}{p{3.8cm}p{5.2cm}}
\toprule
\textbf{Parameter} & \textbf{Configuration} \\
\midrule

Shared actor backbone
& MLP $6\to128\to128\to64$, ReLU \\

Critic network
& MLP $73\to128\to128\to64\to1$, ReLU \\

Recurrent module
& GRU hidden-state update \\

Discount factor
& $\gamma=0.99$ \\

GAE parameter
& $\lambda=0.95$ \\

PPO clipping threshold
& $\epsilon=0.2$ \\

PPO update epochs
& $K=8$ \\

Actor learning rate
& $3.0\times10^{-4}$ \\

Critic learning rate
& $1.0\times10^{-3}$ \\

LLM checkpoint
& \textbf{Qwen3.5-9B} \\

Quantization
& GGUF Q4\_K\_M \\

Inference backend
& llama-cpp-python \\

Context length
& 4096 tokens \\

Maximum output tokens
& 384 \\

Temperature
& $0.2$ \\

Top-$p$
& $0.95$ \\

Sampling top-$k$
& 40 \\

Repetition penalty
& $1.1$ \\

LLM random seed
& 42 \\

GPU-offloaded layers
& 20 \\

Output format
& JSON object \\

LLM invocation interval
& Every 1 training episode \\

Caching
& Disabled \\

Deployment-time LLM calls
& 0 \\

\bottomrule
\end{tabular}%
}
\end{table}



\section{Performance Evaluation}
This section evaluates the proposed closed-loop framework from two perspectives. First, we verify the effectiveness of the CPFL-enabled V2I DT reconstruction method, including the geometry-aware environmental representation, multi-task weight sensitivity, and PL/KPI prediction accuracy. Second, we evaluate how the reconstructed V2I DT supports the dual-stage E2E latency optimization, including DT-assisted V2I association, joint UL and DL transmission latency optimization control via the unified SAGE-MAPPO framework.

Experiments\footnote{The simulation code for all experiments in this work can be found at https://github.com/nihaoeveryone404/anonymous-v2i-dt.} were conducted in PyTorch with an NVIDIA RTX 4070 GPU. The detailed CPFL and SAGE-MAPPO hyperparameter settings are summarized in Tables~I and~II, respectively. For the DT reconstruction task, we compare the proposed CPFL-based PFL-MTL method with three  representative FL baselines: FedAvg-MTL  \cite{mcmahan2017communication}, FedFly \cite{fedfly2023}, and
FedGBDT \cite{li2020fedgbdt}. FedAvg-MTL denotes standard federated MTL without personalized heads. FedFly represents a communication-efficient FL baseline, while FedGBDT denotes a personalized FL method based on gradient boosting. The proposed method adopts a shared backbone with personalized PL/KPI heads and compressed shared backbone's transmission. The joint sparsification and quantization function $\mathcal Q$ in Eq. (\ref{eq:delta}) employs $p_s = 37$\% and $n_q=8$.

For the V2I association and E2E latency optimization task, we compare three baseline strategies. Baseline~1 connects each vehicle to the strongest-SINR MBS:
\begin{equation}
w_{i,b}^{(\mathrm{B1})}
=
\begin{cases}
1, & b=\arg\max_{b'\in\mathcal B_c}\gamma^{\mathrm{DL}}_{b',i},\\
0, & \mathrm{otherwise},
\end{cases}
\label{eq:baseline1_assoc}
\end{equation}
while baseline~2 adopts a rate-proportional soft association strategy:
\begin{equation}
\begin{aligned}
w_{i,b}^{(\mathrm{B2})}
=
\frac{
\left(V_{i,b}\right)^{\alpha}
}{
\sum_{b'\in\mathcal B_i}
\left(V_{i,b}\right)^{\alpha}
}.
\label{eq:baseline2_assoc}
\end{aligned}
\end{equation}
The proposed strategy further incorporates the CPFL-predicted KPI-aware utility and LoS confidence to construct a DT-assisted soft multi-MBS association, which is then used by the unified LLM-MARL controller for dual E2E latency optimization.

\begin{figure}[!t]
\centering
\captionsetup[subfloat]{font=scriptsize, skip=2pt}

\subfloat[MBS Deployment]{
    \label{fig_dataset_bs}
    \includegraphics[width=0.48\columnwidth]{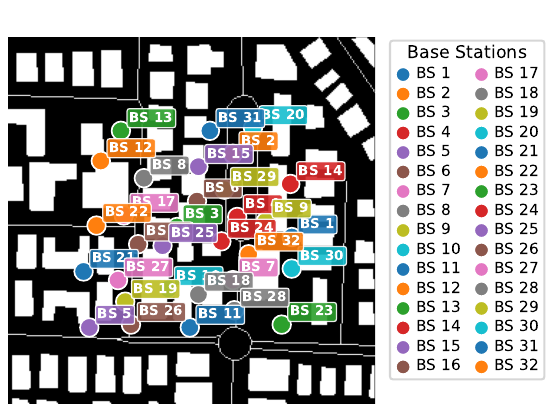}
}
\hfill 
\subfloat[LoS Confidence Map]{
    \label{fig_dataset_los}
    \includegraphics[width=0.44\columnwidth]{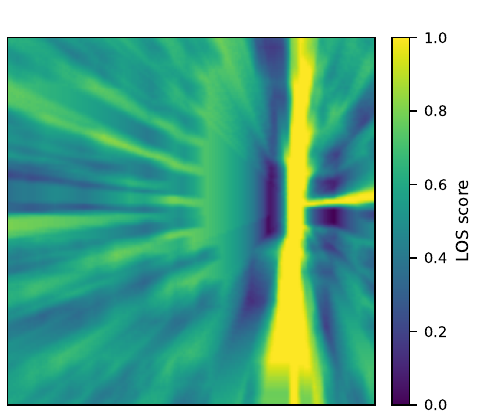}
}

\caption{Overview of the Physical Environment Representations.}
\label{fig_dataset}
\end{figure}

\subsection{Performance of V2I DT Reconstruction}
Fig.~3 illustrates the physical environment representation used by the CPFL-enabled V2I DT. Fig.~3(a) shows the deployment of MBSs over the road network, while Fig.~3(b) presents the reconstructed LoS-confidence map. The LoS map provides a geometry-aware visibility prior, which helps distinguish stable LoS links from deeply blocked NLoS links. Therefore, the proposed DT is not merely a data-driven PL/KPI predictor; it also embeds physically meaningful blockage and visibility information into the subsequent association and latency-optimization stages.

\begin{table}[!t]
\vspace{0.3cm}
\caption{Performance with Different Multi-Task Weight $\alpha$.}
\vspace{-0.2cm}
\label{tab:alpha_tuning}
\centering
\resizebox{0.99\columnwidth}{!}{
\begin{tabular}{@{}llccccc@{}}
\toprule
\textbf{Scenario} & \textbf{Metric} & \textbf{$\alpha$=0.2} & \textbf{$\alpha$=0.3} & \textbf{$\alpha$=0.4} & \textbf{$\alpha$=0.5} & \textbf{$\alpha$=0.6} \\ \midrule
\multirow{2}{*}{Low-load} & PL (RMSE) & 7.21 & 6.83 & 5.89 & 5.72 & 6.01 \\
                          & KPI (RMSE) & 0.078 & 0.062 & 0.044 & 0.042 & 0.048 \\ \midrule
\multirow{2}{*}{Mid-load} & PL (RMSE) & 8.12 & 7.58 & 6.57 & 6.51 & 6.83 \\
                          & KPI (RMSE) & 0.069 & 0.057 & 0.041 & 0.039 & 0.043 \\ \midrule
\multirow{2}{*}{High-load}& PL (RMSE) & 9.35 & 8.76 & 7.90 & 7.98 & 8.32 \\
                          & KPI (RMSE) & 0.086 & 0.073 & 0.052 & 0.051 & 0.057 \\ \midrule
\multirow{2}{*}{All}      & PL (RMSE) & 6.78 & 6.35 & 5.48 & 5.36 & 5.68 \\
                          & KPI (RMSE) & 0.054 & 0.045 & 0.029 & 0.022 & 0.028 \\ \bottomrule
\end{tabular}
}
\end{table}

\begin{table}[!t]
\caption{PL and KPI Prediction Performance Comparison.}
\vspace{-0.2cm}
\label{tab:pl_performance}
\centering
\resizebox{0.99\columnwidth}{!}{
\begin{tabular}{@{}ccccc@{\hspace{0.35cm}}ccccc@{}}
\toprule
\textbf{Scen.} & \textbf{Task} & \textbf{Strategy} & \textbf{RMSE} & \textbf{MAE}
&
\textbf{Scen.} & \textbf{Task} & \textbf{Strategy} & \textbf{RMSE} & \textbf{MAE} \\
\midrule

\multirow{8}{*}{\shortstack[c]{Low-\\load}}
& PL  & FedAvg-MTL & 5.5909 & 4.8641
&
\multirow{8}{*}{\shortstack[c]{High-\\load}}
& PL  & FedAvg-MTL & 7.8508 & 6.9087 \\

& PL  & FedFly     & 4.2153 & 3.5408
&
& PL  & FedFly     & 4.9321 & 4.2909 \\

& PL  & FedGBDT     & 4.8912 & 4.3042
&
& PL  & FedGBDT     & 6.5412 & 5.6254 \\

& PL  & \textbf{Ours}    & \textbf{3.0682} & \textbf{2.6079}
&
& PL  & \textbf{Ours}    & \textbf{3.0860} & \textbf{2.5922} \\

& KPI & FedAvg-MTL & 0.0403 & 0.0334
&
& KPI & FedAvg-MTL & 0.0486 & 0.0422 \\

& KPI & FedFly     & 0.0298 & 0.0253
&
& KPI & FedFly     & 0.0275 & 0.0233 \\

& KPI & FedGBDT     & 0.0351 & 0.0305
&
& KPI & FedGBDT     & 0.0392 & 0.0337 \\

& KPI & \textbf{Ours}    & \textbf{0.0214} & \textbf{0.0181}
&
& KPI & \textbf{Ours}    & \textbf{0.0181} & \textbf{0.0150} \\

\midrule

\multirow{8}{*}{\shortstack[c]{Mid-\\load}}
& PL  & FedAvg-MTL & 6.3648 & 5.6010
&
\multirow{8}{*}{All}
& PL  & FedAvg-MTL & 5.2073 & 4.4782 \\

& PL  & FedFly     & 5.8231 & 5.0661
&
& PL  & FedFly     & 3.8451 & 3.2298 \\

& PL  & FedGBDT     & 6.1045 & 5.3109
&
& PL  & FedGBDT     & 4.5120 & 3.9254 \\

& PL  & \textbf{Ours}    & \textbf{5.5909} & \textbf{4.6963}
&
& PL  & \textbf{Ours}    & \textbf{3.0023} & \textbf{2.5519} \\

& KPI & FedAvg-MTL & 0.0371 & 0.0315
&
& KPI & FedAvg-MTL & 0.0192 & 0.0163 \\

& KPI & FedFly     & 0.0234 & 0.0198
&
& KPI & FedFly     & 0.0105 & 0.0089 \\

& KPI & FedGBDT     & 0.0318 & 0.0273
&
& KPI & FedGBDT     & 0.0153 & 0.0131 \\

& KPI & \textbf{Ours}    & \textbf{0.0151} & \textbf{0.0125}
&
& KPI & \textbf{Ours}    & \textbf{0.0063} & \textbf{0.0051} \\

\bottomrule
\end{tabular}
}
\end{table}

Table~III reports the sensitivity of the multi-task weight $\alpha$.
Across Low-, Mid-, High-, and All-load scenarios, the best PL/KPI trade-off consistently appears around $\alpha\in[0.4,0.5]$. For example, in the All-load scenario, the PL Root Mean Square Error (RMSE) decreases from $6.78$ at $\alpha=0.1$ to $5.36$ at
$\alpha=0.5$, while the KPI RMSE decreases from $0.054$ to $0.022$. Similar trends can be observed in the other load regimes. This demonstrates that KPI supervision acts as an effective regularizer for PL learning. When $\alpha$ is too small, the model over-emphasizes KPI fitting and weakens the geometric representation required for accurate PL reconstruction. When $\alpha$ is too large, the model becomes overly PL-oriented and loses sensitivity to service-related KPI dynamics. Based on this observation, we adopt $\alpha^\star=0.45$ in the final implementation, which lies near the robust Pareto-optimal region.

We further compare the PL and KPI prediction performance of different
FL schemes, which is shown in Table~IV. The proposed method achieves the lowest RMSE and MAE for both PL and KPI prediction in almost all scenarios.  These results verify that the shared backbone can capture common propagation and traffic patterns, while the personalized heads effectively absorb vehicle- and environment-specific Non-IID characteristics. Compared with FedFly and FedGBDT, the proposed method also shows clear advantages. FedFly improves communication efficiency, but its shared representation is less effective in modeling heterogeneous PL/KPI structures. FedGBDT introduces personalization, but its prediction accuracy is still limited in the coupled PL and KPI reconstruction task. In contrast, the proposed CPFL framework benefits from personalized local heads, multi-task supervision, and compressed shared backbone's transmission.

\subsection{Performance of Closed-Loop V2I DT-Assisted Dual E2E Latency Optimization}







\begin{figure}[t]
\centering
\captionsetup[subfloat]{font=scriptsize}
\subfloat[V2I Association Topology]{
\includegraphics[width=0.95\columnwidth]{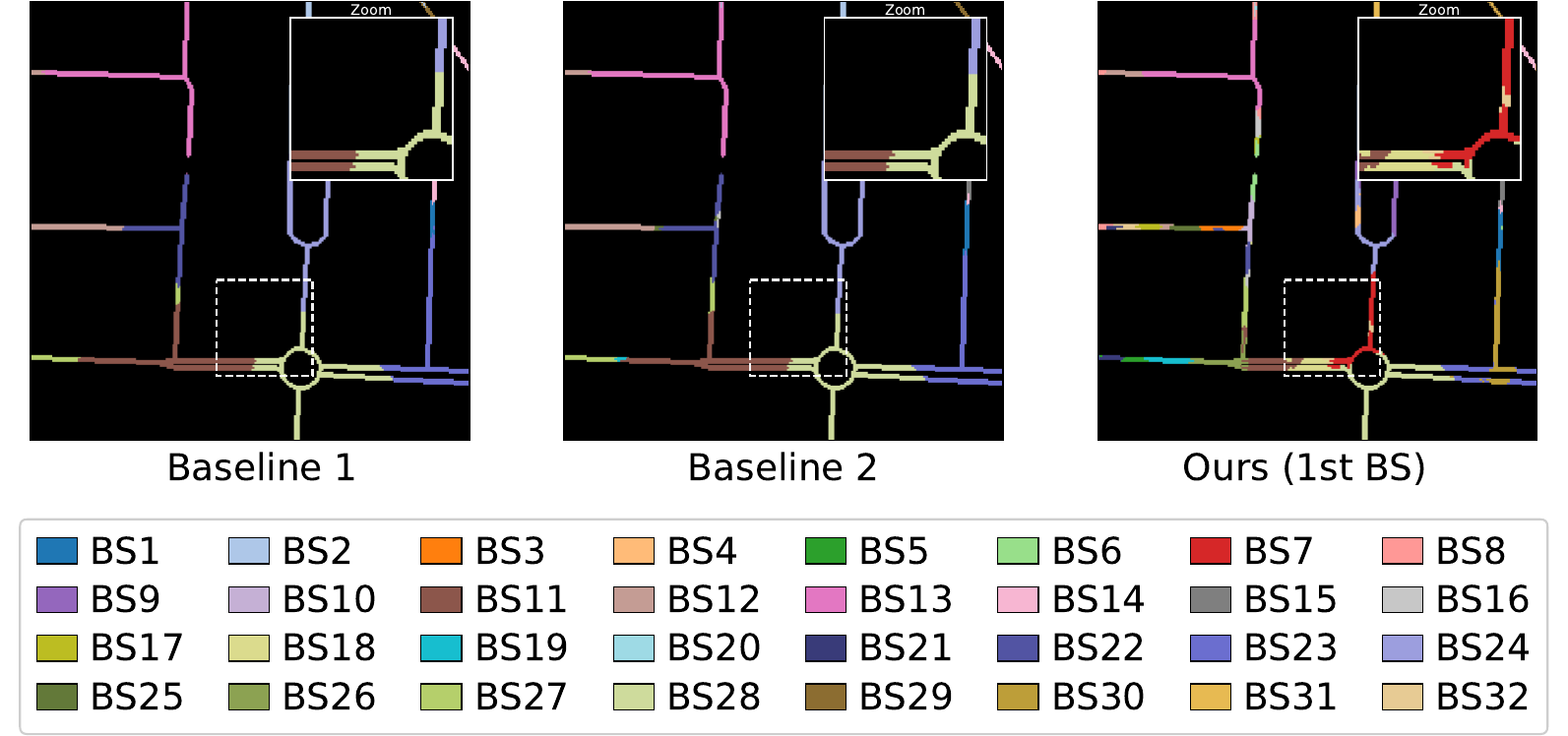}}
\vspace{-0.1mm}
\subfloat[V2I Association Topology (Ours)]{
\includegraphics[width=0.95\columnwidth]{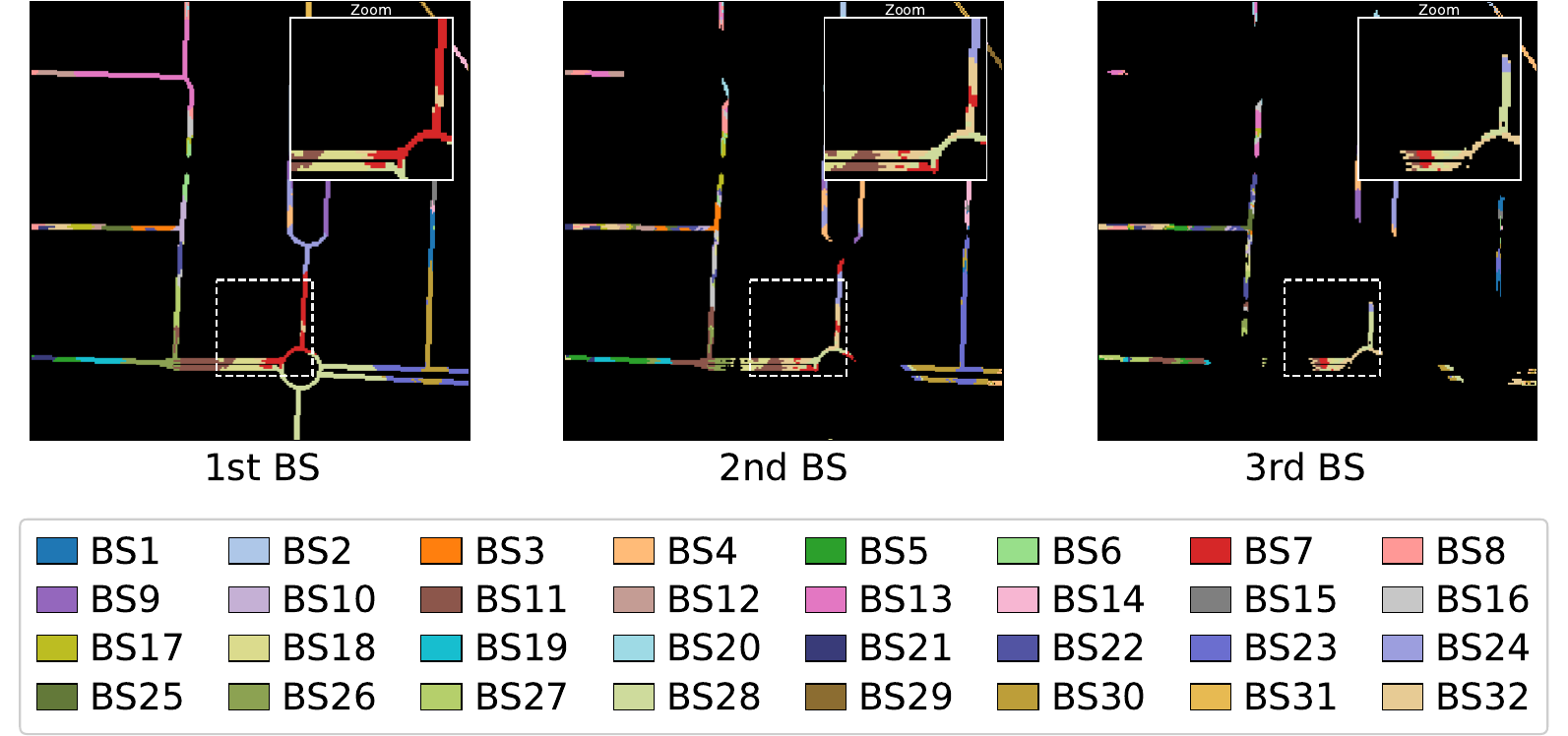}}
\vspace{-0.1cm}
\caption{DT-Assisted V2I Association Visualization. }
\vspace{0.2cm}
\label{fig_bs_selection}
\end{figure}

\begin{figure}[!t]
\centering
\captionsetup[subfloat]{font=scriptsize}
\subfloat[Downlink Throughput Heatmap]{
\includegraphics[width=0.95\columnwidth]{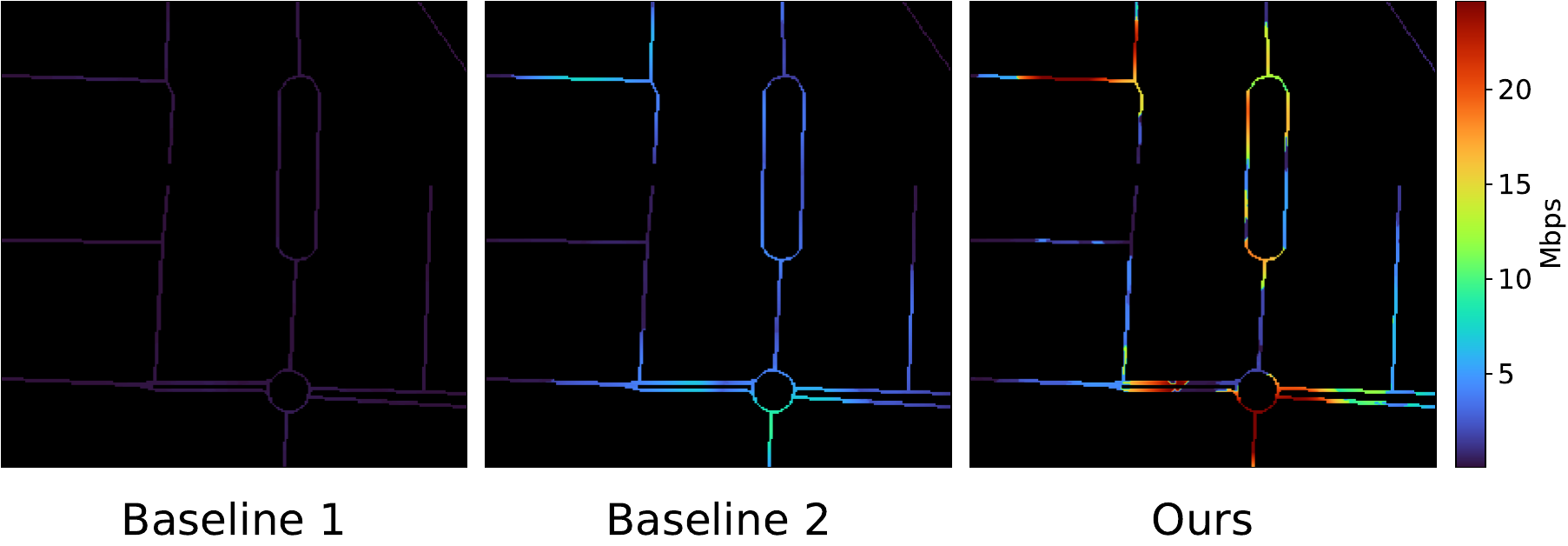}}
\vspace{-0.1mm}
\subfloat[Average Queuing Delay Heatmap]{
\includegraphics[width=0.95\columnwidth]{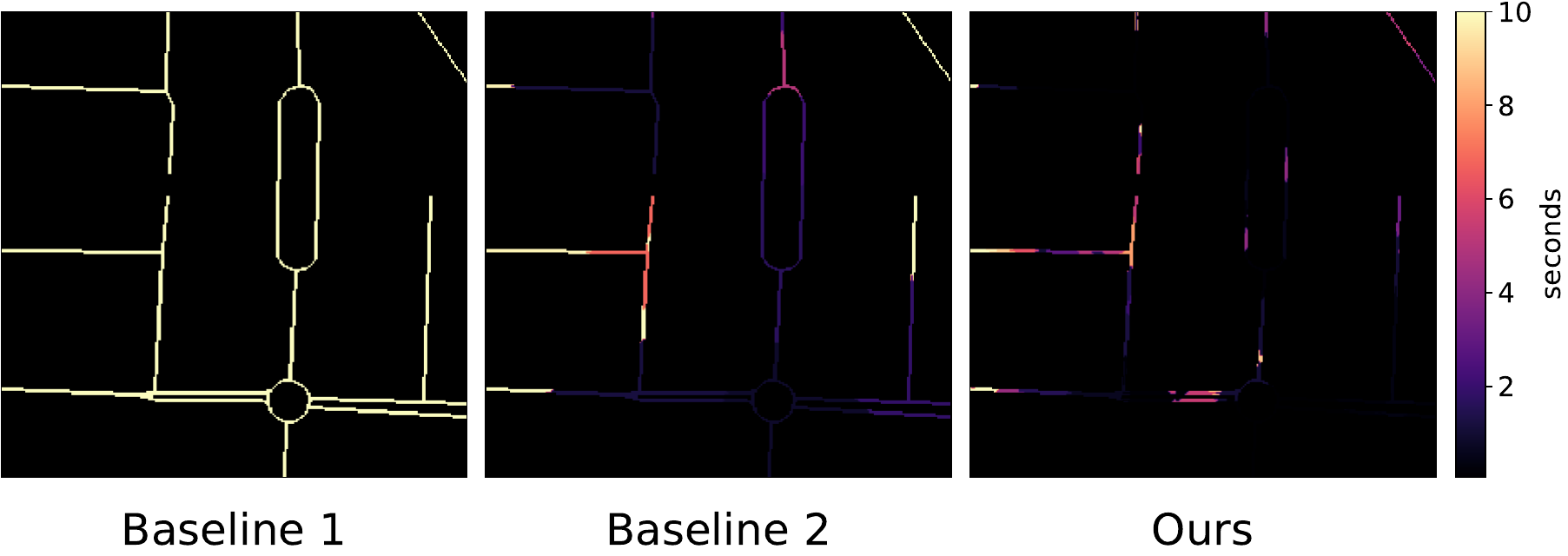}}
\vspace{-0.1mm}
\subfloat[QoE Heatmap]{
\includegraphics[width=0.95\columnwidth]{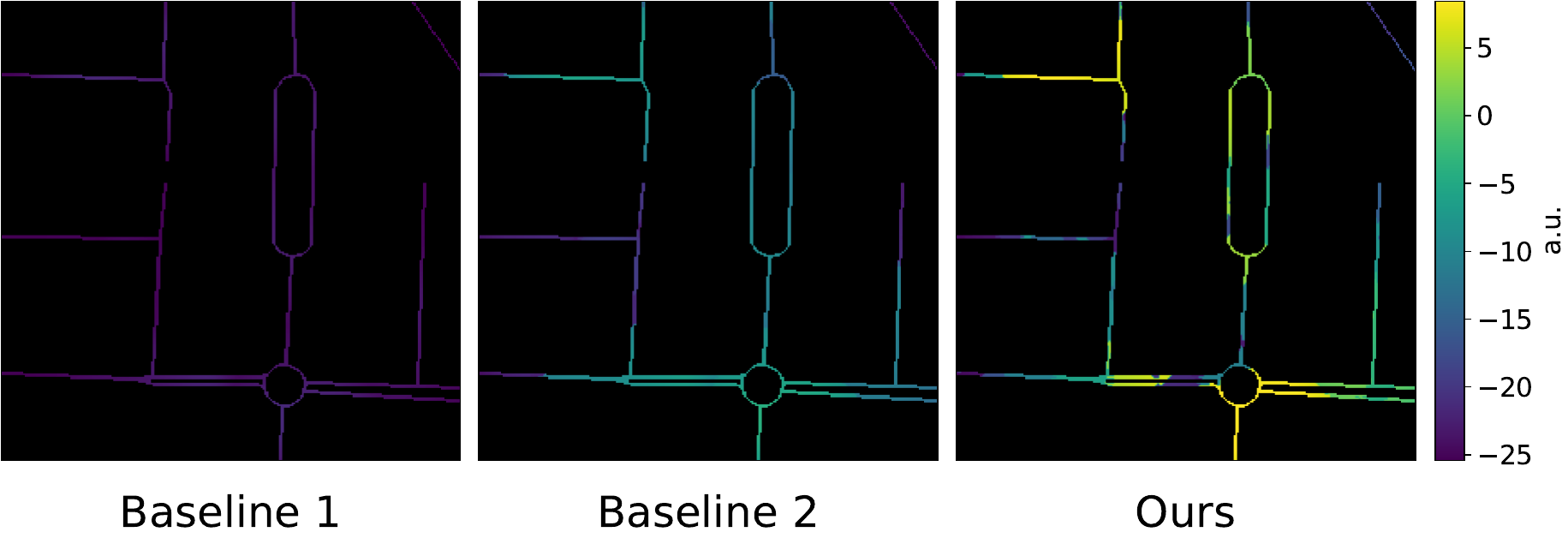}}
\vspace{-0.1cm}
\caption{V2I Association Performance Comparison.}
\end{figure}

As Fig.~4 shows, we visualize the V2I association behavior under different strategies.
Baseline~1 produces rigid serving regions because each vehicle is associated with the strongest-SINR MBS. This rule can overload a few strong MBSs and generate severe queueing hotspots. Baseline~2 softens the association through utility-proportional weighting, but it still lacks explicit geometric awareness. In contrast, the proposed method generates more adaptive serving regions by jointly exploiting the CPFL-predicted KPI-aware utility and LoS confidence. The zoomed regions further show that the proposed association strategy can adjust boundaries around intersections, blocked road segments, and transition regions, thereby improving robustness against blockage and load imbalance. Fig.~5 further compares the spatial distributions of DL throughput, average queueing delay, and QoE. In the throughput heatmap, the proposed method fills more low-throughput road segments and expands the high-rate service region. In the queueing-delay heatmap, baseline~1 exhibits obvious congestion hotspots, while our proposed method significantly suppresses high-delay regions by distributing traffic across multiple MBSs. In the QoE heatmap, the proposed method produces a broader high-QoE serving region, confirming that the throughput improvement and delay reduction are jointly translated into better user experience.


\begin{figure}[t]
\centering
\hspace{-1.8em}
\begin{minipage}[t]{0.48\columnwidth}
\centering
\includegraphics[width=\linewidth]{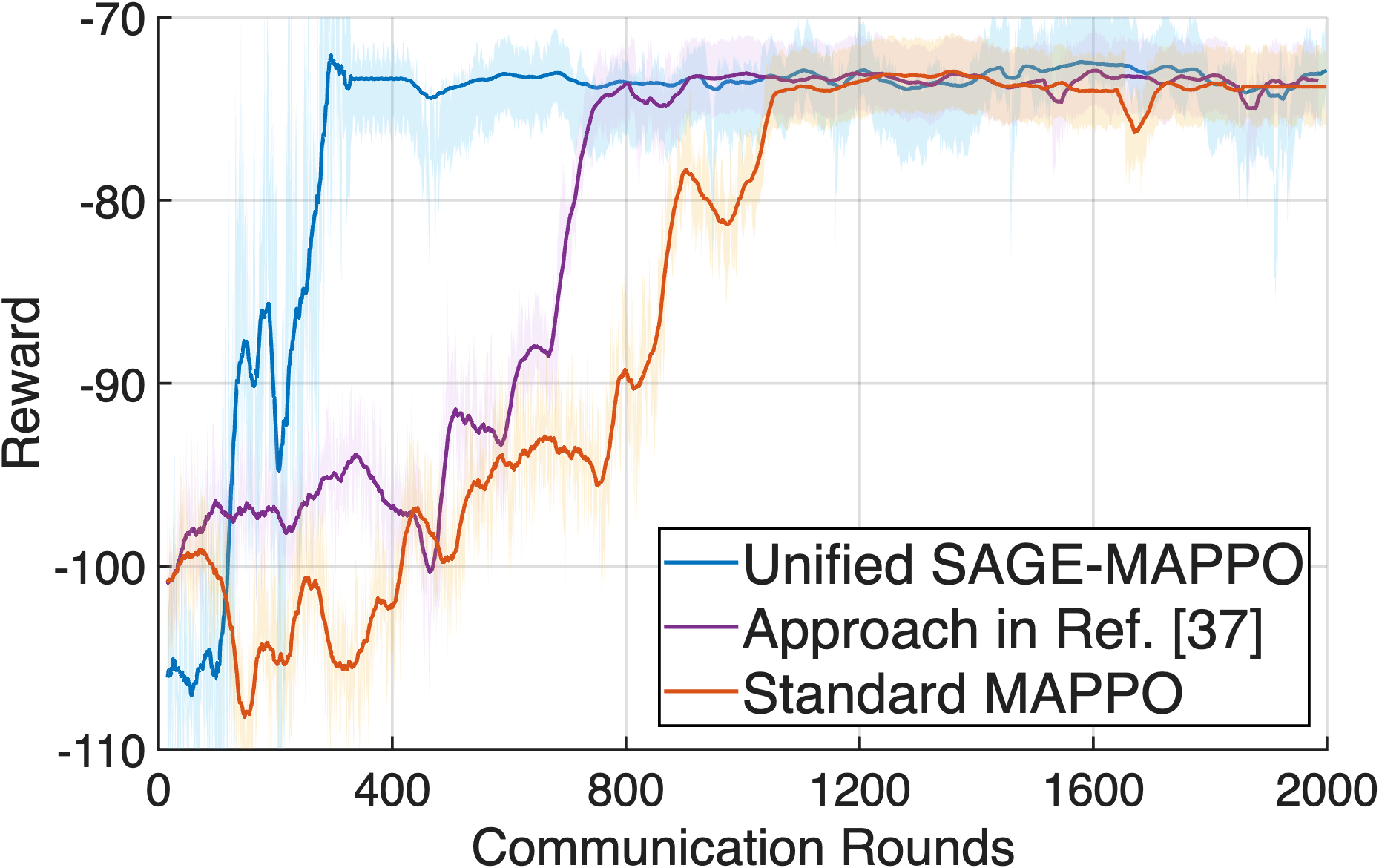}
\subfloat{\small (a) Problem 1}
\end{minipage}
\begin{minipage}[t]{0.48\columnwidth}
\centering
\includegraphics[width=\linewidth]{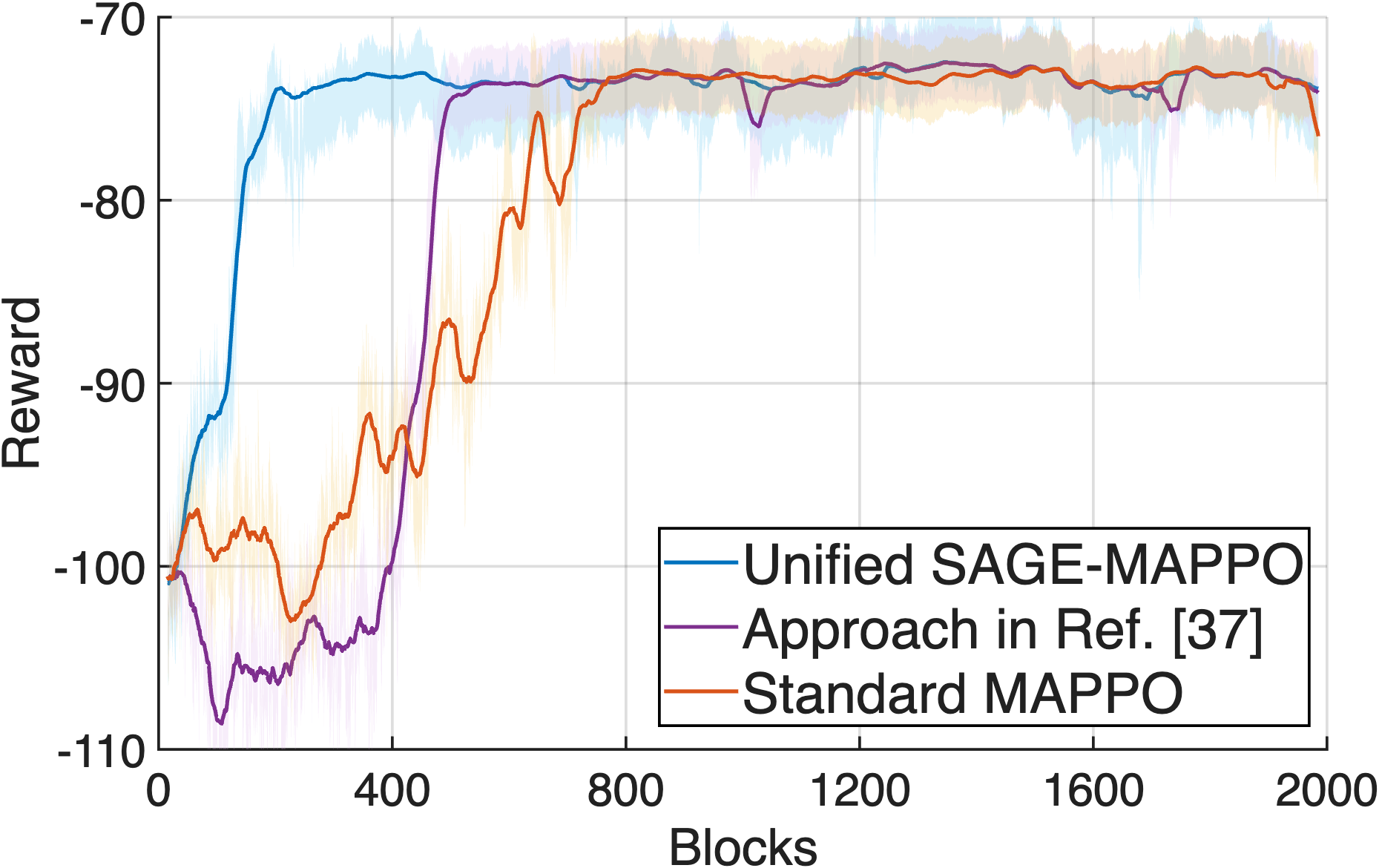}
\subfloat{\small (b) Problem 2}
\end{minipage}
\caption{Reward Comparison.}
\end{figure}

We then evaluate the joint UL and DL latency performance of the
proposed unified SAGE-MAPPO framework for both problems. Fig.~6 compares the reward evolution of the proposed method with the accelerated MARL approach in
Ref.~\cite{wei2025optimizing} and standard MAPPO for both problems within the first updated DT. In Fig.~6(a), the proposed unified SAGE-MAPPO rapidly increases the reward and reaches a stable high-reward region much earlier than the two baselines, which shows its effectiveness in learning the joint UL and DL shared backbone transmission policy for Problem~1. The approach in Ref.~\cite{wei2025optimizing} also improves the convergence speed over standard MAPPO, but still requires more communication rounds to approach the same reward level. Standard MAPPO converges the slowest and exhibits more pronounced fluctuations due to inefficient exploration in the coupled traffic splitting and power control
space. A similar trend is observed in Fig.~6(b) for Problem~2. The proposed method quickly converges within the early training blocks and maintains stable performance afterwards, indicating that the same SAGE-MAPPO architecture can be effectively adapted from CPFL parameter transmission optimization to
service-aware E2E latency optimization. The convergence advantage comes from the LLM-based semantic guidance and the shared policy structure between the two problems, which together reduce redundant exploration and provide a warm start across the two optimization stages. 

\begin{figure}[t]
\centering
\vspace{0.2cm}
\begin{minipage}[t]{0.48\columnwidth}
\centering
\includegraphics[width=\linewidth]{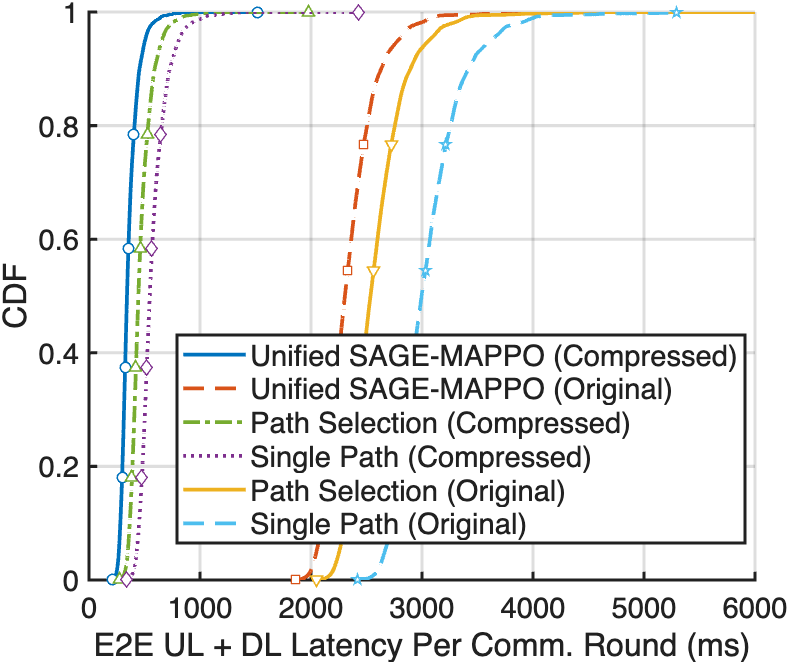}
\subfloat{\small (a) Problem 1}
\end{minipage}
\begin{minipage}[t]{0.48\columnwidth}
\centering
\includegraphics[width=\linewidth]{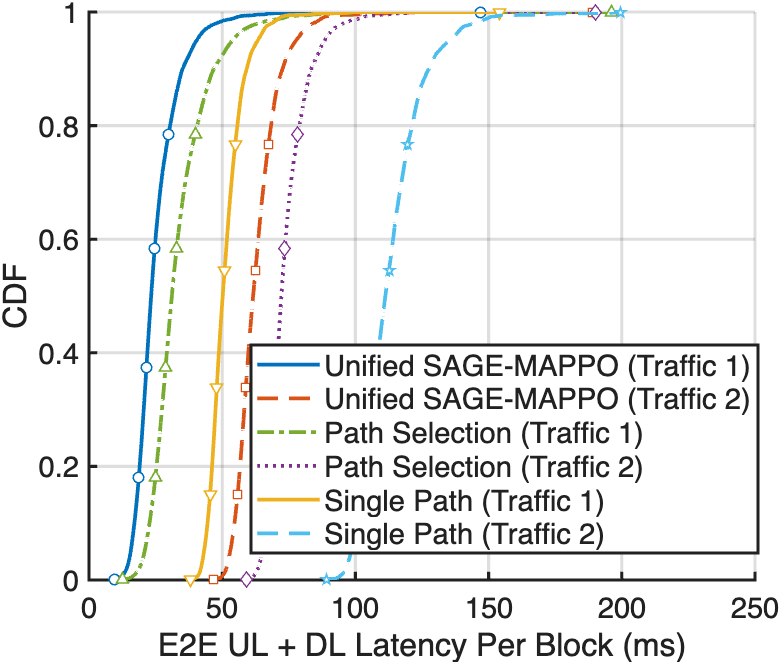}
\subfloat{\small (b) Problem 2}
\end{minipage}
\caption{The CDF Comparison of E2E joint UL and DL Latency.}
\vspace{0.2cm}
\end{figure}

Meanwhile, Fig.~7 presents the E2E latency Cumulative Distribution Function (CDF) comparison for both problems. For Problem~1, Fig.~7(a) shows that the proposed unified SAGE-MAPPO with compressed shared backbone's transmission achieves the
left-most CDF curve, indicating the lowest CPFL parameter transmission latency. Compared with path-selection and single-path schemes, the proposed method reduces both median and tail latency by jointly optimizing UL traffic splitting, UL transmit power, and DL traffic splitting over multiple MBS paths. The clear gap between the compressed and original schemes further confirms that shared backbone's compression effectively reduces model transmission latency. For Problem~2, Fig.~7(b) shows that the proposed unified SAGE-MAPPO also shifts the service-level latency CDFs to the left for both service types. This demonstrates that the proposed controller improves average and tail E2E latency by balancing multi-MBS traffic splitting and UL power allocation. In contrast, the single-path baseline suffers from larger latency because it cannot exploit path diversity and is more vulnerable to congested or weak links. Overall, Fig.~7 verifies that the proposed framework reduces both CPFL shared backbone's transmission latency and service-aware E2E latency.

\begin{figure}[h]
\centering
\vspace{0.2cm}
\begin{minipage}[t]{0.48\columnwidth}
\centering
\includegraphics[width=\linewidth]{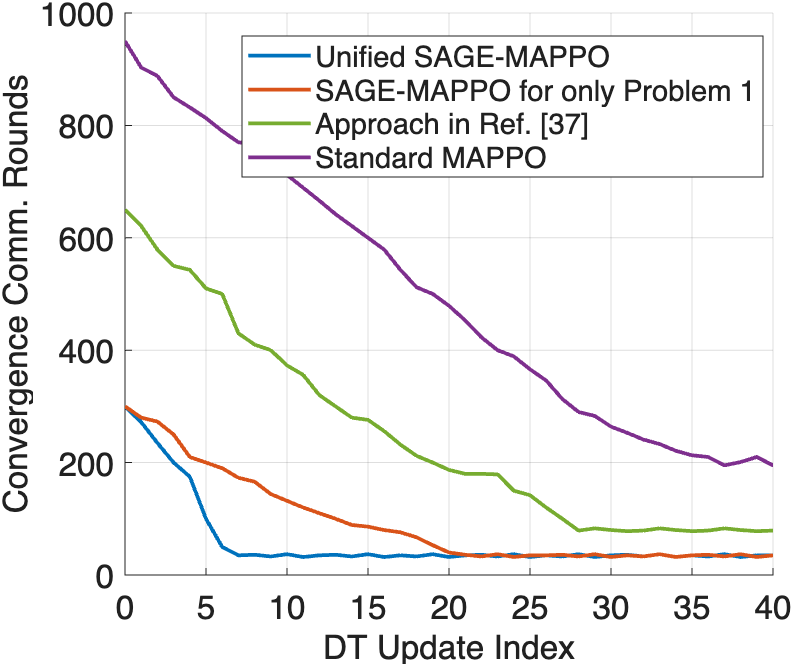}
\subfloat{\small (b) Problem 1}
\end{minipage}
\begin{minipage}[t]{0.48\columnwidth}
\centering
\includegraphics[width=\linewidth]{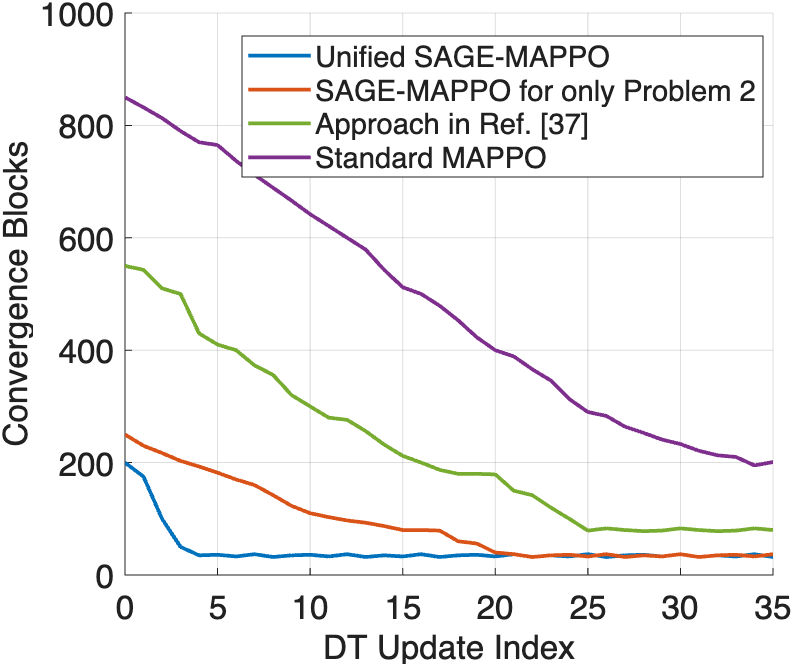}
\subfloat{\small (a) Problem 2}
\end{minipage}
\caption{Convergence Trend in terms of  DT Update Index.}
\vspace{0.2cm}
\end{figure}

We further analyze the convergence of the unified SAGE-MAPPO framework from a closed-loop DT-update perspective. As shown in Fig.~8, we compare the number of communication rounds required for convergence in Problem~1 and the number of blocks required for convergence in Problem~2 under different DT update indices. For Problem~1, the unified SAGE-MAPPO converges within far fewer communication rounds than SAGE-MAPPO trained only for Problem~1, the accelerated MARL approach in Ref.~\cite{wei2025optimizing}, and standard MAPPO~\cite{10026882}. For
Problem~2, a similar trend can be observed: the unified SAGE-MAPPO requires fewer convergence blocks than the problem-specific SAGE-MAPPO, the approach
in Ref.~\cite{wei2025optimizing}, and standard MAPPO~\cite{10026882}. This confirms that the unified architecture benefits not only from LLM-guided semantic supervision, but also
from cross-problem policy reuse. The convergence gain comes from the fact that Problem~1 and Problem~2 share
the same agent roles and action structure, i.e., vehicle-side UL traffic
splitting and power control together with server/MBS-side DL traffic
splitting. This closed-loop training mechanism reduces
redundant exploration and accelerates convergence compared with training two
independent MARL architectures.

\section{Conclusion}
In this paper, we proposed a closed-loop dual-stage E2E latency optimization framework for DT-assisted 6G V2I systems. A CPFL-enabled decision-oriented DT was developed to jointly reconstruct PL and 
KPI-related states while retaining raw vehicle measurements and personalized heads locally. Within each DT update cycle, two distinct but temporally coupled problems were formulated: Problem 1 minimizes 
the worst-vehicle E2E joint UL and DL latency of compressed shared backbone exchange during DT reconstruction, whereas Problem 2 minimizes the aggregate service-aware E2E joint UL and DL latency after the reconstructed DT is available. The two problems are operationally decoupled through their execution stages and timescales, while their aligned agent roles and control structures enable cross-problem policy reuse. To solve them efficiently, we developed a unified SAGE-MAPPO framework, in which the LLM acts as a training-time semantic planner to reduce inefficient exploration, while MARL agents execute feasible low-level actions. Simulation results demonstrate the improved PL/KPI prediction accuracy, more effective DT-assisted V2I association, reduced median and tail latency, and faster convergence than the considered MARL baselines. Future work 
will investigate online validation using real-world mobility traces, 
robustness against imperfect DT predictions and semantic guidance, and adaptive DT-update scheduling.


\appendix
\appendices
\subsection{Proof of Lemma 1}
\begin{proof}
According to Assumption 1, let $x = x_i^{t,e}$, and $y = x_i^{t,e+1}$ and thus $ x_i^{t,e+1} - x_i^{t,e}
=
\bigl(
-\eta_\mathbf{w} g_{i,\mathbf{w}}(x_i^{t,e}),
-\eta_\mathbf{v} g_{i,\mathbf{v}}(x_i^{t,e})
\bigr)
$. Then we have
\begin{equation}
\begin{aligned}
\mathbb{E}\!\left[ f_i(x_i^{t,e+1})\mid x_i^{t,e}\right] 
&\le
f_i(x_i^{t,e})
+
\left\langle
\nabla f_i(x_i^{t,e}),
x_i^{t,e+1} - x_i^{t,e}
\right\rangle
\\&+
\frac{L}{2}
\left\|
x_i^{t,e+1} - x_i^{t,e}
\right\|^2 
\\&=
f_i(x_i^{t,e})
-
\eta_\mathbf{w}
\left\langle
\nabla_\mathbf{w} f_i(x_i^{t,e}),
g_{i,\mathbf{w}}(x_i^{t,e})
\right\rangle
\\& -
\eta_\mathbf{v}
\left\langle
\nabla_\mathbf{v} f_i(x_i^{t,e}),
g_{i,\mathbf{v}}(x_i^{t,e})
\right\rangle
\\&
+
\frac{L}{2}
\left(
\eta_\mathbf{w}^2 \| g_{i,\mathbf{w}}(x_i^{t,e}) \|^2
+
\eta_\mathbf{v}^2 \| g_{i,\mathbf{v}}(x_i^{t,e}) \|^2
\right).
\end{aligned}
\end{equation}


According to Assumption 2, we have $\mathbb{E}[g_{i,\mathbf{w}}] = \nabla_\mathbf{w} f_i$,
$\mathbb{E}[g_{i,\mathbf{v}}] = \nabla_\mathbf{v} f_i$, $\mathbb{E}\!\left[
\| g_{i,\mathbf{w}} - \nabla_\mathbf{w} f_i \|^2
\right]
\le
\sigma^2$, and $\mathbb{E}\!\left[
\| g_{i,\mathbf{v}} - \nabla_\mathbf{v} f_i \|^2
\right]
\le
\sigma^2$.
As
$\mathbb{E}\!\left[\|g\|^2\right]
=
\|\mathbb{E}[g]\|^2
+
\mathbb{E}\!\left[\|g - \mathbb{E}[g]\|^2\right]$,  we obtain $\mathbb{E}\!\left[
\| g_{i,\mathbf{w}}(x) \|^2
\right]
\le
\| \nabla_\mathbf{w} f_i(x) \|^2 + \sigma^2$ and $\mathbb{E}\!\left[
\| g_{i,\mathbf{v}}(x) \|^2
\right]
\le
\| \nabla_\mathbf{v} f_i(x) \|^2 + \sigma^2$.
Taking the total expectation of $\mathbb{E}\!\left[ f_i(x_i^{t,e+1})\mid x_i^{t,e}\right]$ in terms of $x_i^{t,e}$, we obtain 
\begin{equation}
\small
\begin{aligned}
f_i\!\left(x_i^{t,e+1}\right)
& \le
f_i\!\left(x_i^{t,e}\right)
-
\left(
\eta_\mathbf{w} - \frac{L}{2}\eta_\mathbf{w}^2
\right)
\left\|
\nabla_\mathbf{w} f_i\!\left(x_i^{t,e}\right)
\right\|^2
\\& -
\left(
\eta_\mathbf{v} - \frac{L}{2}\eta_\mathbf{v}^2
\right)
\left\|
\nabla_\mathbf{v} f_i\!\left(x_i^{t,e}\right)
\right\|^2
+
\frac{L}{2}
\left(
\eta_\mathbf{w}^2 + \eta_\mathbf{v}^2
\right)
\sigma^2 .
\end{aligned}
\label{eq:appendix 1_1}
\end{equation}
Let the step sizes satisfy $\eta_\mathbf{w}, \eta_\mathbf{v} \leq \frac{1}{L}$, then $\eta_\mathbf{w} - \frac{L}{2}\eta_\mathbf{w}^2 \ge \frac{\eta_\mathbf{w}}{2}$ and $\eta_\mathbf{v} - \frac{L}{2}\eta_\mathbf{v}^2 \ge \frac{\eta_\mathbf{v}}{2}$. Then Eq. (\ref{eq:appendix 1_1}) can be written as 
\begin{equation}
\begin{aligned}
f_i(x_i^{t,e+1})
&\le
f_i(x_i^{t,e})
-
\frac{1}{2}
(
\eta_\mathbf{w} \| \nabla_\mathbf{w} f_i(x_i^{t,e}) \|^2
\\& +
\eta_\mathbf{v} \| \nabla_\mathbf{v} f_i(x_i^{t,e}) \|^2
) +
\frac{L}{2}
\left(
\eta_\mathbf{w}^2 + \eta_\mathbf{v}^2
\right)
\sigma^2 .
\end{aligned}
\end{equation}
\end{proof}
\vspace{-1cm}
\subsection{Proof of Lemma 2}
\begin{proof}
For the communication round $t$, the local drift for vehicle $i$ at step $e$ is $\mathbf{w}_i^{t,e+1}
=
\mathbf{w}_i^{t,e}
-
\eta_\mathbf{w}\, g_{i,\mathbf{w}}$.
Thus the local drift for $E$ steps within a whole round is $\mathbf{w}_i^{t,e+1}
=
\mathbf{w}_i^{t,e}
-
\eta_\mathbf{w}\sum_{e=0}^{E-1}\, g_{i,\mathbf{w}}$.
Then
\begin{equation}
\begin{aligned}
\bigl\| \mathbf{w}_i^{t,E} - \mathbf{w}^{t} \bigr\|^{2}
 &=
\eta_\mathbf{w}^{2}
\left\|
\sum_{e=0}^{E-1}
g_{i,\mathbf{w}}^{e}
\right\|^{2} \\& \leq \eta_\mathbf{w}^{2}E\sum_{e=0}^{E-1}
\left\|
g_{i,\mathbf{w}}^{e}
\right\|^{2}\\& \leq \eta_\mathbf{w}^{2}E\sum_{e=0}^{E-1}
\left(
\| \nabla_\mathbf{w} f_i(x) \|^2 + \sigma^2
\right).
\end{aligned}
\label{eq:lemma2_1}
\end{equation}
According to Assumption 4, we have
\begin{equation}
\begin{aligned}
\frac{1}{|\mathcal I|}
\sum_{i=1}^{|\mathcal I|}
\mathbb{E}\!\left[
\left\|
\nabla_\mathbf{w} f_i\!\left(x_i^{t,e}\right)
\right\|^{2}
\right]
\;\le\;
\mathbb{E}\!\left[
\left\|
\nabla_\mathbf{w} F\!\left(x^{t,e}\right)
\right\|^{2}
\right]
+\gamma^{2}.
\end{aligned}
\label{eq:lemma2_2}
\end{equation}
Therefore, according to Eq. (\ref{eq:lemma2_1}) and Eq. (\ref{eq:lemma2_2}), we have
\begin{equation}
\begin{aligned}
\frac{1}{|\mathcal I|}\sum_{i=1}^{|\mathcal I|} 
\mathbb{E}\!\left[
\bigl\| \mathbf{w}_i^{t,E} - \mathbf{w}^t \bigr\|^2
\right]
\le
C\,\eta_\mathbf{w}^2 E^2 \left( \sigma^2 + \gamma^2 \right).
\end{aligned}
\end{equation}
\end{proof}

\vspace{-0.8cm}
\subsection{Proof of Theorem 1}
\begin{proof}
From communication round $t$ to $t+1$ with $E$ steps, $\mathbf{w}^{t+1} = \frac{1}{|\mathcal I|}\sum_{i} \mathbf{w}_i^{t,E}$ and $\mathbf{v}_i^{t+1} = \mathbf{v}_i^{t,E}$. Then
\begin{equation}
\small
\begin{aligned}
f_i(\mathbf{w}^{t+1}, \mathbf{v}_i^{t,E}) - f_i(\mathbf{w}^{t}, \mathbf{v}_i^{t})
&=
\bigl[
f_i(\mathbf{w}_i^{t,E}, \mathbf{v}_i^{t,E}) - f_i(\mathbf{w}^{t}, \mathbf{v}_i^{t})
\bigr]
\\&+
\bigl[
f_i(\mathbf{w}^{t+1}, \mathbf{v}_i^{t,E}) - f_i(\mathbf{w}_i^{t,E}, \mathbf{v}_i^{t,E})
\bigr].
\end{aligned}
\end{equation}
Thus we have
\begin{equation}
\small
\begin{aligned}
F^{t+1} - F^{t}
&=
\underbrace{
\frac{1}{|\mathcal I|}
\sum_{i}
\bigl[
f_i(x_i^{t,E}) - f_i(x_i^{t,0})
\bigr]
}_{A}
\\&+
\underbrace{
\frac{1}{|\mathcal I|}
\sum_{i}
\bigl[
f_i(\mathbf{w}^{t+1}, \mathbf{v}_i^{t,E})
-
f_i(\mathbf{w}_i^{t,E}, \mathbf{v}_i^{t,E})
\bigr]
}_{B} .
\end{aligned}
\label{eq:theorem1_1}
\end{equation}

According to Lemma~1, summing  over $e = 0, \ldots, E-1$ and averaging over vehicles $i$, we obtain
\begin{equation}
\small
\begin{aligned}
\mathbb{E}(A) &
\le
-
\frac{\eta_\mathbf{w}}{2}
\sum_{e=0}^{E-1}
\mathbb{E}\!\left[
\frac{1}{|\mathcal I|}
\sum_{i}
\bigl\|
\nabla_\mathbf{w} f_i\!\left(x_i^{t,e}\right)
\bigr\|^{2}
\right]
\\& -
\frac{\eta_\mathbf{v}}{2}
\sum_{e=0}^{E-1}
\mathbb{E}\!\left[
\frac{1}{|\mathcal I|}
\sum_{i}
\bigl\|
\nabla_\mathbf{v} f_i\!\left(x_i^{t,e}\right)
\bigr\|^{2}
\right]
+
\frac{L E}{2}
\left(
\eta_\mathbf{w}^{2}
+
\eta_\mathbf{v}^{2}
\right)
\sigma^{2}.
\end{aligned}
\label{eq:theorem1_2}
\end{equation}

Under Assumption 1, in terms of the shared backbone $\mathbf{w}$, we have
\begin{equation}
\begin{aligned}
f_i\!\left(\mathbf{w}^{t+1}, \mathbf{v}_i^{t,E}\right)
&\le
f_i\!\left(\mathbf{w}_i^{t,E}, \mathbf{v}_i^{t,E}\right)
\\& +
\left\langle
\nabla_\mathbf{w} f_i\!\left(\mathbf{w}_i^{t,E}, \mathbf{v}_i^{t,E}\right),
\, \mathbf{w}^{t+1} - \mathbf{w}_i^{t,E}
\right\rangle
\\& +
\frac{L}{2}
\left\|
\mathbf{w}^{t+1} - \mathbf{w}_i^{t,E}
\right\|^{2},
\end{aligned}
\end{equation}
where $\mathbf{w}^{t+1} = \frac{1}{N}\sum_{i} \mathbf{w}_i^{t,E}$ and $\mathbf{v}_i^{t+1} = \mathbf{v}_i^{t,E}$. While
\begin{equation}
\small
\begin{aligned}
\mathbb{E}(B)
& \le
\mathbb{E}\!\left[
\frac{1}{|\mathcal I|}
\sum_{i}
\langle \nabla_\mathbf{w} f_i\!\left(\mathbf{w}_i^{t,E}, \mathbf{v}_i^{t,E}\right), \mathbf{w}^{t+1} - \mathbf{w}_i^{t,E} \rangle
\right]
\\& +
\frac{L}{2}
\mathbb{E}\!\left[
\frac{1}{|\mathcal I|}
\sum_{i}
\left\|
\mathbf{w}^{t+1} - \mathbf{w}_i^{t,E}
\right\|^{2}
\right] \\& \leq
\frac{\eta_\mathbf{w}}{4}
\mathbb{E}\!\left[
\frac{1}{|\mathcal I|}
\sum_{i}
\left\|
f_i\!\left(\mathbf{w}_i^{t,E}, \mathbf{v}_i^{t,E}\right)
\right\|^{2}
\right]
\\& +
\left(
\frac{1}{\eta_\mathbf{w}}
+
\frac{L}{2}
\right)
\mathbb{E}\!\left[
\frac{1}{|\mathcal I|}
\sum_{i}
\left\|
\mathbf{w}_i^{t,E} - \mathbf{w}^{t+1}
\right\|^{2}
\right]\\&
\leq \frac{\eta_\mathbf{w}}{4}
\sum_{e=0}^{E-1}
\mathbb{E}\!\left[
\left\|
\nabla_\mathbf{w} F^{t}
\right\|^{2}
\right]
+
CL
\mathbb{E}\!\left[
\frac{1}{|\mathcal I|}
\sum_{i}
\left\|
\mathbf{w}_i^{t,E} - \mathbf{w}^{t+1}
\right\|^{2}
\right]\\&
\leq \frac{\eta_\mathbf{w}E}{4}
\mathbb{E}\!\left[
\left\|
\nabla_\mathbf{w} F^{t}
\right\|^{2}
\right]
+
CL
\mathbb{E}\!\left[
\frac{1}{|\mathcal I|}
\sum_{i}
\left\|
\mathbf{w}_i^{t,E} - \mathbf{w}^{t}
\right\|^{2}
\right].
\end{aligned}
\label{eq:theorem1_3}
\end{equation}
Let $\eta_\mathbf{w} \leq \frac{c}{LE}$. According to Lemma~2, combine Eq. (\ref{eq:theorem1_1})-(\ref{eq:theorem1_3}), we have
\begin{equation}
\begin{aligned}
\mathbb{E}\!\left[F^{t+1}\right]
& \le
\mathbb{E}\!\left[F^{t}\right]
-
\frac{\eta_\mathbf{w} E}{4}
\mathbb{E}\!\left[
\left\|
\nabla_\mathbf{w} F^{t}
\right\|^{2}
\right]
\\& -
\frac{\eta_\mathbf{v}}{4}
\sum_{e=0}^{E-1}
\mathbb{E}\!\left[
\frac{1}{|\mathcal I|}
\sum_{i}
\left\|
\nabla_\mathbf{v} f_i\!\left(x_i^{t,e}\right)
\right\|^{2}
\right]
\\& +
C_1 L E
\left(
\eta_\mathbf{w}^{2}
+
\eta_\mathbf{v}^{2}
\right)
\sigma^{2}
+
C_2 L \eta_\mathbf{w}^{2} E^{2} \gamma^{2}.
\end{aligned}
\label{eq:theorem1_4}
\end{equation}
Summing Eq. (\ref{eq:theorem1_4}) from $t=0$ to $T-1$, we obtain
\begin{equation}
\begin{aligned}
\mathbb{E}[F^{0}] - \mathbb{E}[F^{T}]
&\ge
\frac{\eta_\mathbf{w} E}{4}
\sum_{t=0}^{T-1}
\mathbb{E}\!\left[
\left\|
\nabla_\mathbf{w} F^{t}
\right\|^{2}
\right]
\\&+
\frac{\eta_\mathbf{v}}{4}
\sum_{t=0}^{T-1}
\sum_{e=0}^{E-1}
\mathbb{E}\!\left[
\frac{1}{|\mathcal I|}
\sum_{i}
\left\|
\nabla_\mathbf{v} f_i\!\left(x_i^{t,e}\right)
\right\|^{2}
\right] \\& - T(C_1 L E
\left(
\eta_\mathbf{w}^{2}
+
\eta_\mathbf{v}^{2}
\right)
\sigma^{2}
+
C_2 L \eta_\mathbf{w}^{2} E^{2} \gamma^{2}).
\end{aligned}
\end{equation}

Under Assumption 3, we further obtain
\begin{equation}
\small
\begin{aligned}
\frac{1}{T}
\sum_{t=0}^{T-1}
\mathbb{E}\!\left[
\left\|
\nabla_\mathbf{w} F^{t}
\right\|^{2}
\right]
& \le
\frac{4\left(F^{0}-F^\star\right)}{\eta_\mathbf{w} E T}
+
C_1 L \eta_\mathbf{w} \sigma^{2}
+
C_2 L \eta_\mathbf{w} E \gamma^{2},
\end{aligned}
\label{eq:bound1}
\end{equation}
and
\begin{equation}
\begin{aligned}
& \frac{1}{TE} \sum_{t=0}^{T-1}\sum_{e=0}^{E-1}
\mathbb{E}\!\left[
\frac{1}{|\mathcal I|} \sum_{i=1}^{|\mathcal I|} \left\lVert \nabla_{\mathbf{v}_i} f_i\!\left(x_i^{t,e}\right) \right\rVert^2
\right]
\le
\frac{4\left(F^0 - F^\star\right)}{\eta_\mathbf{v} T E }
\\&+
\frac{4}{\eta_\mathbf{v} E }(C_1 L E
\left(
\eta_\mathbf{w}^{2}
+
\eta_\mathbf{v}^{2}
\right)
\sigma^{2}
+
C_2 L \eta_\mathbf{w}^{2} E^{2} \gamma^{2}).
\end{aligned}
\label{eq:bound2}
\end{equation}
Combining the two bounds in Eq. (\ref{eq:bound1}) and Eq. (\ref{eq:bound2}) and rearranging them in terms of $G^{t}$ yields the main theorem.
\end{proof}

\bibliographystyle{IEEEtran}
\bibliography{IEEEabrv,ref}

\end{document}